\documentclass[journal,10pt]{IEEEtran}

\usepackage[T1]{fontenc}
\usepackage{amssymb}
\usepackage[utf8]{inputenc}
\usepackage[english]{babel}
\usepackage{amsmath}				
\usepackage{comment}				
\usepackage[T1]{fontenc}
\usepackage{mathtools}
\usepackage{algcompatible,amsmath}
\usepackage[ruled,vlined]{algorithm2e}
\usepackage{graphicx}
\usepackage{multicol}
\usepackage{caption}
\usepackage{afterpage}
\usepackage[noend]{algpseudocode}
\usepackage{color}
\usepackage{etoolbox}
\usepackage{url}
\usepackage{subcaption,siunitx,booktabs}
\usepackage{siunitx,booktabs}
\usepackage{placeins}
\usepackage{bbm}
\usepackage{soul}
\usepackage[noadjust]{cite}

\newcommand\Bigger[2][7]{\left#2\rule{0mm}{#1truemm}\right.}

\normalsize

\makeatletter
\patchcmd{\@makecaption}
{\scshape}
{}
{} 
{}
\makeatother
\usepackage{subcaption}
\usepackage{graphicx}
\usepackage{amsthm}

\newtheorem{proposition}{Proposition}
\newtheorem{theorem}{Theorem}
\newtheorem*{proof*}{Proof}

\newtheorem{remark}{Remark}				

\newcommand{\norm}[1]{\left\lVert#1\right\rVert}

\IEEEoverridecommandlockouts
\makeatletter
\let\origIEEEPARstart\IEEEPARstart
\renewcommand{\IEEEPARstart}[3][1.1]{%
	\def\@IEEEPARstartDROPDEPTH{#1\baselineskip}%
	\origIEEEPARstart{#2}{#3}%
}
\makeatother

\ifCLASSINFOpdf

\else

\fi

\begin{document}

\title{Resource Reservation in Backhaul and Radio Access Network with Uncertain User Demands}

\author{Navid~Reyhanian,
	Hamid~Farmanbar, and Zhi-Quan Luo,~\IEEEmembership{Fellow, IEEE}
	\thanks{N. Reyhanian is with the Department
		of Electrical and Computer Engineering, University of Minnesota, Minneapolis,
		MN, 55455 USA (e-mail: navid@umn.edu).}
	\thanks{H. Farmanbar is with Huawei Canada Research Center, Ottawa, Canada (e-mail: hamid.farmanbar@huawei.com).}
	\thanks{Z.-Q. Luo is with Shenzhen Research Institute of Big Data, The Chinese University of Hong Kong, Shenzhen, China (e-mail: luozq@cuhk.edu.cn).}
	\thanks{This paper was presented in part at the $21^{\text{st}}$ IEEE International Workshop on Signal Processing Advances in Wireless Communications (SPAWC), Atlanta, GA, USA, May 26--29, 2020 \cite{Reyhanian1}.}}


\maketitle

\begin{abstract}
Resource reservation is an essential step to enable wireless data networks to support a wide range of user demands. In this paper, we consider the problem of joint resource reservation in the backhaul and Radio Access Network (RAN) based on the statistics of user demands and channel states, and also network availability. The goal is to maximize the sum of expected traffic flow rates, subject to link and access point budget constraints, while minimizing the expected outage of downlinks. The formulated problem turns out to be non-convex and difficult to solve to global optimality. We propose an efficient Block Coordinate Descent (BCD) algorithm to approximately solve the problem. The proposed BCD algorithm optimizes the link capacity reservation in the backhaul using a novel multi-path routing algorithm that decomposes the problem down to link-level and parallelizes the computation across backhaul links, while the reservation of transmission resources in RAN is carried out via a novel scalable and distributed algorithm based on Block Successive Upper-bound Minimization (BSUM). We prove that the proposed BCD algorithm converges to a Karush–Kuhn–Tucker (KKT) solution. Simulation results verify the efficiency and the efficacy of our BCD approach against two heuristic algorithms.
\end{abstract}


\IEEEpeerreviewmaketitle

\section{Introduction}\label{sec:intro}
Resource reservation is an important step in network planning and management due to its significant effects on the user quality of service.
For wireless data networks operating in random and dynamic environments, finding resource reservation protocols that remain robust under uncertain user demands is challenging.
Resource reservation, which balances network performance and its hardware costs, involves traffic forecasting and resource allocation for the predicted traffic \cite{albasheir2015enhanced,kaur2015novel,van2001time}. 
Resource reservation in the backhaul and Radio Access Network (RAN) should satisfy a wide range of applicable traffic demands. In particular, both the link capacity in the backhaul and transmission resources in RAN should be sliced and reserved for users such that upon the arrival of a new demand, the network is able to support it.

Resource reservation for the uncertain demand was first studied by Gomory and Hu in \cite{gomory1962application}, which reserved link capacities using a single commodity routing problem with a finite number of sources.  For communication networks, where both link budget and node budget are to be reserved, different approaches are proposed for resource reservation. In traffic oblivious approaches, to make reservations and slice the network resources, user demand and its statistics are not considered in the problem formulation \cite{5504139,kumar2016kulfi,cicconetti2007end}. The drawback of traffic oblivious approaches is  that they limit the ability of a network to adapt to any given demand. To reserve link capacities in flow networks, a collection of predicted demand scenarios are considered in \cite{applegate2006making,moehle2017distributed}. The proposed algorithms in \cite{applegate2006making,moehle2017distributed} reserve link capacities such that the predicted demand scenarios are supported as much as possible. The accuracy of the reservations in \cite{applegate2006making,moehle2017distributed} is based on the number of predicted scenarios. However, as the number of scenarios increases, the complexity of solving the problem increases. Short term user demands are predicted by Long Short-Term Memory (LSTM) neural networks in \cite{8528330,yan2019intelligent,9109582}. Recurring resource reservations based on the short-term traffic variations incur reconfiguration costs, service interruptions, and overhead in networks \cite{khan2020network}.
The mean of user demands is used in \cite{prados2019complete} to balance the workload among a
set of data centers in a network that consists of the backhaul and RAN such that the utilization of resources is maximized. The joint reservation of computational and radio resources is studied in \cite{8265188}, where different ranges are considered for uncertain user demands. A linear program is formulated in \cite{8265188} to support the uncertain user demands, which vary in given ranges, as much as the network allows. In \cite{9096559}, the transmission resource reservation in RAN is considered where the minimum requirements of users are known and deterministic. The authors of \cite{9096559} proposed a matching-based algorithm to solve an optimization problem with the goal of minimizing the consumption of network resources while meeting the requirements of users.

Optimal routing is studied widely for many settings, e.g., \cite{bertsekas1991linear,tsitsiklis1986distributed,bertsekas1998network}, while optimal resource allocation in RAN has also been studied for different wireless channels, e.g., \cite{yan2019intelligent,9057480,li2001capacity,liao2019model,sciancalepore2019rl,liang2005gaussian,kim2011optimal,9096559}.
The joint routing in the backhaul network and resource allocation RAN is studied in a number of more recent papers \cite{liao2014base,xiao2004simultaneous,el2014joint,wang2018joint,9107209,liu2019optimizing,kordbacheh2019robust,karakayali2007joint}. In \cite{liao2014base} and \cite{9107209}, the user demand requirements are deterministic and known. On the other hand, in \cite{xiao2004simultaneous,wang2018joint,karakayali2007joint,liu2019optimizing,kordbacheh2019robust}, the traffic of users is maximized as much as the network is able to support, regardless of user demand statistics. To find a robust resource reservation, network resources should be reserved based on demand statistics.
In \cite{xiao2004simultaneous,el2014joint,wang2018joint} and \cite{karakayali2007joint,liu2019optimizing,kordbacheh2019robust}, the wireless channel capacity is a deterministic function of input power. Moreover, the convexity of the problem is assumed in \cite{xiao2004simultaneous,el2014joint,wang2018joint,kordbacheh2019robust}. Neither of these assumptions holds in practice, where the wireless channel capacity is random and its distribution is a function of supplied transmission resources \cite{4411539,1210731,4626300}.  

In addition to different proposed formulations for resource allocations and network planning with certain and uncertain user demands in existing literature, several algorithms have been used to solve the resulting optimization problems. Among them, the Alternating Direction Method of Multipliers (ADMM) has been used widely \cite{moehle2017distributed,liao2018distributed,liao2014base,nguyen2017parallel,wu2019toward}. ADMM enables flow decoupling in the network optimization process.
The efficiency of ADMM depends on the number of auxiliary link variables introduced to make the optimization subproblems separable. For networks with a large number of links, ADMM can be slow, i.e., requiring a large number of iterations. 
A dual decomposition method for path-based routing is used in \cite{1664999}, where a gradient ascent approach has been proposed to solve the dual problem. Since in most problems the dual function is non-smooth, the gradient ascent approach has to take small steps, resulting in slow convergence. A distributed approach for large-scale revenue management problems in airline networks is proposed by Kemmer \textit{et al.} in \cite{kemmer2012dynamic}. The single-path dynamic programming approach in \cite{kemmer2012dynamic} has shown great success in practice despite the absence of convergence or solution enhancement guarantees.

In this paper, we propose a  resource management scheme for end-to-end resource reservation, i.e., from data centers to users, based on user demand and downlink achievable rate statistics for a data network consisting of the backhaul and RAN. We consider a multi-path routing in our formulation, where a user can be served by several Access Points (APs) through multiple paths from a data center. We formulate the problem of jointly reserving the transmission resources in RAN and link capacities in the backhaul based on user demand and downlink achievable rate statistics so as to maximize the total expected supportable user traffic, while minimizing the expected outage of downlinks. Since the formulated problem is non-convex and hard to solve,  we propose an efficient Block Coordinate Descent (BCD) algorithm, which is convergent to a Karush–Kuhn–Tucker (KKT) solution of the resource reservation problem. 

In the proposed BCD approach, one block of variables determines the link capacity reservation in the backhaul and the other block of variables specifies the transmission resource reservation in RAN. We alternately optimize the two blocks of variables in the  BCD algorithm. Fixing the transmission resources in RAN, we update the link capacity reservation in the backhaul  via a novel multi-path routing algorithm. Inspired by  the resource level decomposition ideas in \cite{kemmer2012dynamic}, the proposed multi-path routing  decomposes the problem down to link-level and parallelizes the computation across backhaul links. 
	Based on the convergence theory for Block Successive Upper-bound Minimization (BSUM) methods in \cite{razaviyayn2013unified}, we prove that the proposed multi-path routing is convergent to the global minima of an arbitrary convex cost function with Lipschitz continuous gradient. The required computation time for each iteration of the proposed multi-path routing is equal to that for one link regardless of the network size. After updating the link capacity reservations,  we update the transmission resource reservation in RAN.
		 Since the resource reservation problem in RAN is possibly non-convex, we propose a distributed algorithm based on 
			the BSUM techniques to iteratively solve a sequence of convex approximations of the original problem. We prove that the proposed BCD algorithm converges to a KKT solution.  To verify the performance of the proposed algorithm, two heuristic algorithms are also developed and used as benchmarks to evaluate the efficiency and the efficacy of the proposed approach via simulations.

The rest of this paper is organized as follows. The system model and problem formulation are given in Section II. Section III describes a general scalable and distributed algorithm for the multi-path flow routing. In Section IV, we propose a BCD algorithm for the network resource reservation problem. The simulation results are given in Section V, and concluding remarks are given in Section VI.

\section{System Model and Problem Formulation}\label{sec:model}

Consider a typical scenario whereby user data is transmitted via backhaul network links from data centers to APs in RAN, which in turn relay the data to the desired users as depicted in Fig. \ref{fig:cran}. Suppose $\mathcal{B}$ denotes the set of APs and $\mathcal{K}$ denotes the set of mobile users. The set of directed wired links of the backhaul is denoted by $\mathcal{L}$. A path connects a data center and an AP through a sequence of wired links in the backhaul and finally goes through one downlink to reach the end user. The downlinks between APs and users are predetermined according to channel quality, interference levels, and path loss.
\begin{figure}
	\centering
	\includegraphics[width=0.4\textwidth]{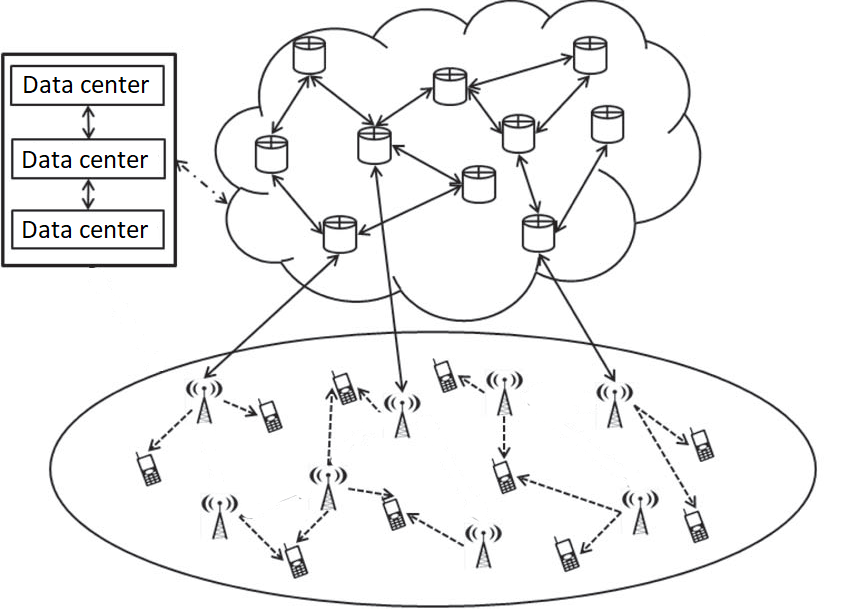}
	\caption{A network comprised of APs and backhaul parts.}
	\label{fig:cran}
\end{figure}

We consider each user demands one commodity and there are $K=|\mathcal{K}|$ datastreams in the backhaul network. The proposed
scheme can be easily extended to the scenario that each user
demands multiple commodities. To serve each user, several candidate paths are selected between the origin and destination, and traffic reservation for the corresponding commodities is implemented over those paths. 
The candidate paths can go through different APs, and the joint transmission of APs to a user (coordinated multi-point mode) is considered in this paper. Only the last hop on each path is wireless.

Each path is denoted by $p$, and the set of all paths is represented by $\mathcal{P}$. The set of paths that carry user $k$ data is denoted by $\mathcal{P}_k$. The backhaul network links comprising path $p$ for serving user $k$ are represented by the set $\mathcal{L}_k^p$. Similarly, the network nodes on path $p\in \mathcal{P}_k$ are denoted by the set $\mathcal{U}_k^p$. The demand of user $k$ is a random variable represented by $d_k$. It follows a certain Probability Density Function (PDF) denoted by $f_k(d_k)$. The corresponding Cumulative Density Function (CDF) is represented by $F_k(d_k)$. Let $r_k$ denote the traffic rate reserved for user $k$. The actual traffic flow of user $k$ supported by the network is a random variable given by
\begin{displaymath}
\min(d_k,r_k)=
\Bigger[8]\{\begin{array}{@{}cl}
r_k,   & \text{if} \:\:\:r_k \leq d_k,\\[3mm]
d_k,              & \text{otherwise}.
\end{array}
\end{displaymath}
We calculate the expected supportable traffic rate for user $k$ as follows:
\begin{align}
\mathbb{E}\left(\min(d_k,r_k)\right)=\int_{0}^{r_k}y_k. f_k(y_k)dy_k+r_k\int_{r_k}^{\infty}f_k(y_k)dy_k.\nonumber
\end{align} 
Since the network is not able to support the demand when it exceeds the reserved rate, we have the minimum in the above expectation. In the first integral, the random demand of user $k$ falls below the reserved rate. In the second integral, the random demand exceeds $r_k$. 

Since a user receives their data from multiple APs, transmission resources should be reserved in multiple APs for the paths available to the user. The resource reservation in the backhaul and RAN is limited by two physical constraints:
\begin{itemize}
	\item The aggregate reserved traffic rate for paths that share a link must not exceed the link capacity. Therefore, we have the following constraint:
	\begin{align}
	\sum_{k=1}^K\sum_{p:\{p\in \mathcal{P}_k,l\in\mathcal{L}_k^p\}}r_k^p \leq C_l, \hspace{.5cm}\forall l \in \mathcal{L},\label{eq:linkcap}
	\end{align}
	where $r_k^p$ is the reserved traffic rate for path $p$ (for serving user $k$). Moreover, the capacity of link $l$ is denoted by $C_l$. Flows on different paths available to one user are treated as separate flows. Thus, we have the inner summation in the above constraint.
	\item The total reserved transmission resources for different paths must not exceed the AP capacity. Hence, we have
	\begin{align}
	\sum_{k=1}^K\sum_{p:\{p\in \mathcal{P}_k,b \in \mathcal{U}_k^p\}}t_k^p \leq C_b, \hspace{.5cm}\forall b \in \mathcal{B},\label{eq:nodecap}
	\end{align}
	where $t_k^p$ is the reserved transmission resources in AP $b$ to transmit incoming data from path $p\in \mathcal{P}_k$ to user $k$. Moreover, the capacity of AP $b$ is denoted by $C_b$.
\end{itemize}
In addition to the above physical constraints, our multi-path model enforces another constraint. Since each datastream originating from a data center splits into a number of sub-flows, we have the following constraint:
\begin{itemize}
	\item The aggregate reserved traffic rate for the different paths which carry data to one user is equal to the reserved rate for that user. Hence, we have the following constraint:
	\begin{align}
	\sum_{p \in \mathcal{P}_k} r_k^p = r_k, \hspace{.5cm}\forall k.\label{eq:split}
	\end{align}
\end{itemize}

In the considered model, we do not make any assumption about the type of the transmission resource. It can be bandwidth, transmission power, or time-slot fraction. Based on the allocated resources, the distribution of the achievable rate of a downlink follows a particular PDF. As only the last hop on each path is wireless, path $p$ uniquely identifies the downlink of the last hop. 
The achievable rate (i.e., instantaneous capacity) of the downlink of path $p$ is random and follows an arbitrary distribution with a PDF represented by $z_k^p(v_k^p,t_k^p)$ and a CDF denoted by $Z_k^p(v_k^p,t_k^p)$. The PDF is a function of two variables: the achievable rate of the downlink, denoted by $v_k^p$, and the allocated transmission resource, denoted by $t_k^p$.  When the achievable rate of a downlink falls below the reserved rate $r_k^p$, some outage is experienced and its amount is $r_k^p-v_k^p$, given that the amount of allocated transmission resources to the downlink is $t_k^p$. The probability that this amount of outage takes place is $z_k^p(v_k^p,t_{k}^{p})$.  In  light of the above arguments, the expected outage of the downlink of path $p$ is obtained as follows:
\begin{align}
\int_{0}^{r_k^p} z_k^p(v_k^p,t_{k}^{p})\:(r_k^p-v_k^p)dv_k^p.\label{eq:out1}
\end{align}
Since the achievable rate is a continuous random variable, we have the above integral. 

In this paper, we aim to maximize the expected traffic of users as much as the network is able to support, while minimizing the expected outage of downlinks. We formulate the following optimization problem to find resource reservations in the backhaul and RAN:\allowdisplaybreaks
\begin{align}\textstyle
	& \underset{\mathbf{r,t}}{\text{max}}
	& & \hspace{-.2cm}\sum_{k=1}^{K}\Big[\mathbb{E}[\min(r_k,d_k)]-\theta_k\hspace{-.15cm}\sum_{p \in \mathcal{P}_k}\hspace{-.15cm}\int_{0}^{r_k^p}\hspace{-.35cm} z_k^p(v_k^p,t_{k}^{p})\:(r_k^p-v_k^p)dv_k^p\Big]\nonumber\\
	& \text{s.t.}
	& &\hspace{-.1cm}\eqref{eq:linkcap},\eqref{eq:nodecap},\eqref{eq:split}, r_k, r_k^p, t_{k}^{p} \geq 0,\hspace{0.5cm} p\in \mathcal{P}_k,\forall k,\label{opt:first}
\end{align}
where $\theta_k:\theta_k\geq 0$ is a coefficient chosen by the system designer that adjusts the priorities of maximizing the expected supportable traffic of user $k$ and the minimization of the aggregate outage of downlinks, which serve user $k$. The two blocks of variables in the above problem are $\mathbf{r}=\{r_k,r_k^p\}_{p\in \mathcal{P}_k, k=1:K}$ and $\mathbf{t}=\{t_k^p\}_{p\in \mathcal{P}_k, k=1:K}$. 
\begin{remark}
	Suppose that multiple paths available to user $k$ share a downlink (the last hop). The aggregate outage of downlinks for serving user $k$ is calculated as follows:
	\begin{align}
	\sum_{w \in \mathcal{W}_k}\int_{0}^{\sum_{p:\{p\in \mathcal{P}_k,w\in p\}}r_k^p}\hspace{-.2cm} z_k^w(v_{k}^w,t_{k}^w)\times(\hspace{-.2cm}\sum_{p:\{p\in \mathcal{P}_k,w\in p\}}\hspace{-.5cm}r_k^p-v_{k}^w)dv_{k}^w\label{eq:mul},
	\end{align}
	where $\mathcal{W}_k$ is the set of downlinks, each denoted by $w$, for serving user $k$. When multiple paths available to user $k$ share a downlink, the above outage is placed in the objective function of \eqref{opt:first} instead of its second term, which includes \eqref{eq:out1}.
\end{remark}

The maximization problem \eqref{opt:first} is not easy to solve to global optimality. The objective function of \eqref{opt:first} is in general not necessarily jointly concave in $\mathbf{r}$ and $\mathbf{t}$ for an arbitrary PDF $z_k^p(v_{k}^p,t_k^p)$. The reason is that $\int_{0}^{r_k^p} \partial^2 z_k^p(v_{k}^p,t_{k}^{p})/(\partial t_{k}^{p})^2\:(r_k^p-v_{k}^p)dv_{k}^p$ is not always non-negative. 
\begin{proposition}
	Given $\mathbf{t}$, the optimization in \eqref{opt:first} becomes concave in $\mathbf{r}$.
\end{proposition}
\begin{proof}
	Fixing $\mathbf{t}$, the objective function is separable in $k$. We find the Hessian with respect to $r_k$ and $\{r_k^p\}_{p\in\mathcal{P}_k}$ for those objective function terms which are associated with user $k$ as follows:
	\begin{align}
	&\mathbf{H}_k=\begin{pmatrix} 
	-f_k(r_k)& 0 &\dots&0\\ 0 & -z_k^1(r_k^1,t_{k}^1)&\dots&0\\\vdots  & \vdots&\ddots&\vdots\\0&0&\dots&-z_k^{|\mathcal{P}_k|}(r_k^{|\mathcal{P}_k|},t_{k}^{|\mathcal{P}_k|})
	\end{pmatrix}.\nonumber
	\end{align}
	The overall Hessian matrix is
	\begin{align}
	&\mathbf{H}=\begin{pmatrix} 
	\mathbf{H}_1&  & \mathbf{0}\\  & \ddots&\\\mathbf{0}& & \mathbf{H}_K
	\end{pmatrix}.\nonumber
	\end{align}
	It is observed that the above matrix is negative semidefinite. Since the constraints of problem \eqref{opt:first} are all affine, it follows that the maximization \eqref{opt:first} is concave with fixed $\textbf{t}$. 
\end{proof}
Separable constraints on $\textbf{r}$ and $\textbf{t}$ in  \eqref{opt:first} motivate the BCD algorithm. It is straightforward to show that with \eqref{eq:mul} instead of \eqref{eq:out1} in the objective function, the optimization in \eqref{opt:first} remains concave in $\textbf{r}$.

\section{Distributed Multi-Path Routing in the Backhaul}\label{sec:alg}
This section is concerned with solving \eqref{opt:first} when $\mathbf{t}$ is kept fixed, and \eqref{opt:first} is converted to the minimization format after multiplying the objective function by $-1$. In particular, we study a general multi-path routing to minimize any convex cost function with a Lipschitz continuous gradient. We develop an algorithm that is dual-based and decomposes the problem down to link-level and parallelizes computations across links of the network. The required computation time for each iteration of the proposed multi-path routing algorithm is equal to that for one link regardless of the network size. This interesting property makes the proposed algorithm appropriate for the online optimization of large networks.

For each datastream in the network, several candidate paths are selected. We assume that each flow can be split into multiple sub-flows.  To formulate the multi-path routing problem, we first assume that the cost function is separable in variables, i.e., $\psi(\mathbf{r})=\sum_{k=1}^{K}\sum_{p\in \mathcal{P}_k}\psi_k^p(r_k^p)$, where each $\psi_k^p(r_k^p)$ is strictly convex. 

The optimization problem for the multi-path flow routing can be written as follows:
\allowdisplaybreaks
\begin{equation}
\begin{aligned}
& \underset{\mathbf{r}}{\min}
& & \sum_{k=1}^{K}\sum_{p\in \mathcal{P}_k}\psi_k^p(r_k^p)\label{opt:ronly}\\
& \text{s.t.}
& & \eqref{eq:linkcap}, r_k^p \geq 0,\hspace{0.5cm} p\in \mathcal{P}_k,\forall k.
\end{aligned}
\end{equation} 
Since typically the number of variables is greater than the number of constraints in the above optimization, solving the problem is easier in the dual domain.
The Lagrangian function for the above problem is
\begin{align}
&L_c(\mathbf{r},\boldsymbol{\mu},\boldsymbol{\phi})=\sum_{k=1}^{K}\sum_{p\in \mathcal{P}_k}\hspace{-.1cm}\psi_k^p(r_k^p)+\sum_{l\in \mathcal{L}}\mu_l(\sum_{k=1}^K\sum_{p:\{p\in \mathcal{P}_k,l\in\mathcal{L}_k^p\}}\hspace{-.7cm}r_k^p - C_l)\nonumber\\
&-\sum_{k=1}^{K}\sum_{p\in \mathcal{P}_k}\phi_k^p r_k^p,\label{eq:lag}
\end{align}
where $\mu_l:\mu_l\geq 0$ is the Lagrange multiplier for the capacity constraint of link $l$, and $\phi_k^p:\phi_k^p\geq 0$ is the Lagrange multiplier for constraint $r_k^p\geq 0$. Furthermore, $\boldsymbol{\mu}=\{\mu_l\}_{l\in\mathcal{L}}$ and $\boldsymbol{\phi}=\{\phi_k^p\}_{p\in \mathcal{P}_k, k=1:K}$.
We find the dual problem of \eqref{opt:ronly} as follows:
\begin{equation}
\begin{aligned}
& \underset{\boldsymbol{\mu},\boldsymbol{\phi}}{\text{max}}
& & \min _{\mathbf{r}}L_c(\mathbf{r},\boldsymbol{\mu},\boldsymbol{\phi})\label{opt:dual1}\\
& \text{s.t.}
& & \boldsymbol{\mu} \geq \mathbf{0},\boldsymbol{\phi}\geq \mathbf{0}.
\end{aligned}
\end{equation} 
For many cost functions, no closed-form solution for $\mathbf{r}=\arg\min _{\mathbf{r}}L_c(\mathbf{r},\boldsymbol{\mu},\boldsymbol{\phi})$ exists. Therefore, commonly, the above problem is solved via a primal-dual method such as ADMM \cite{moehle2017distributed,liao2018distributed,nguyen2017parallel,wu2019toward}. However, the auxiliary link variables introduced to make the per-flow subproblems of optimization in \eqref{opt:ronly} separable can slow down ADMM in practice. 

Resource level decomposition for large-scale single-path applications was first proposed in \cite{kemmer2012dynamic} to solve the revenue management problems in airline networks. The proposed decomposition in \cite{kemmer2012dynamic} does not involve any auxiliary variables. In spite of the absence of convergence or solution enhancement guarantees, the resource level decomposition has been rather successful in practice. We leverage resource level decomposition ideas to develop a distributed algorithm to solve the general multi-path routing problem \eqref{opt:ronly} in a parallel fashion such that the traffic passing on each link can be obtained independently from the other links. Unlike the dynamic programming approach in \cite{kemmer2012dynamic}, an optimization-based approach is proposed here to solve subproblems.
In each iteration, the proposed dual algorithm decomposes the problem in \eqref{opt:dual1} and solves the subproblems globally and in parallel. The optimized $\mu_l$ in the $j^{\text{th}}$ iteration of the proposed algorithm is denoted by $\mu_l^j$. Here, we explain the decomposition. The dualized link capacity constraints $\sum_{l\in \mathcal{L}}\mu_l(\sum_{k=1}^K\sum_{p:\{p\in \mathcal{P}_k,l\in\mathcal{L}_k^p\}}r_k^p - C_l)$ in the Lagrangian \eqref{eq:lag} are separable across links. Each link $l$ receives $\mu_l(\sum_{k=1}^K\sum_{p:\{p\in \mathcal{P}_k,l\in\mathcal{L}_k^p\}}r_k^p - C_l)$. In each iteration, based on  $\boldsymbol{\mu}^{j-1}=\{\mu_l^{j-1}\}_{l\in \mathcal{L}}$ in the previous iteration, we decompose the non-separable terms in the Lagrangian \eqref{eq:lag},  which include $r_k^p$, across links on path $p$. Each link $l$ of path $p$ receives a portion of
	\begin{align}
	\alpha_{k,l}^{p,j}=\mu_l^{j-1}/\sum_{l'\in \mathcal{L}_k^p}\mu_{l'}^{j-1},\label{eq:alpha}
	\end{align} 
	In the $j^{\text{th}}$ iteration, the decomposed per-link Lagrangian function is as follows:
	\begin{align}
	&L_l(\mathbf{r}_l,\mu_l,\boldsymbol{\phi}_l,\boldsymbol{\mu}^{j-1})=\sum_{k=1}^K\sum_{p:\{p\in \mathcal{P}_k,l\in\mathcal{L}_k^p\}}\alpha_{k,l}^{p,j}\psi_k^p(r_k^p)\nonumber\\
	&+\mu_l(\sum_{k=1}^K\sum_{p:\{p\in \mathcal{P}_k,l\in\mathcal{L}_k^p\}}\hspace{-.5cm}r_k^p - C_l)-\sum_{k=1}^K\sum_{p:\{p\in \mathcal{P}_k,l\in\mathcal{L}_k^p\}}\hspace{-.4cm}\alpha_{k,l}^{p,j}\phi_k^p r_k^p,\label{eq:decomlag}
	\end{align}
	where $\mathbf{r}_l=\{r_k^p\}_{p\in \mathcal{P}_k,l \in \mathcal{L}_k^p, k=1:K}$ and $\boldsymbol{\phi}_l=\{\phi_k^p\}_{p\in \mathcal{P}_k,l \in \mathcal{L}_k^p, k=1:K}$. We notice that based on \eqref{eq:alpha}, $\{\alpha_{k,l}^{p,j}\}_{l\in \mathcal{L}_k^p}$ in \eqref{eq:decomlag} is calculated using $\boldsymbol{\mu}^{j-1}$. Based on the above decomposition, we obtain
\begin{align}
L_c(\mathbf{r},\boldsymbol{\mu},\boldsymbol{\phi})=\sum_{l\in \mathcal{L}}L_l(\mathbf{r}_l,\mu_l,\boldsymbol{\phi}_l,\boldsymbol{\mu}^{j-1}). \label{eq:equality}
\end{align}
Instead of solving the problem in \eqref{opt:dual1}, we solve
\begin{equation}
\begin{aligned}
&\hspace{-.4cm} \underset{\{\mu_l,\boldsymbol{\phi}_l\}_{l\in \mathcal{L}}}{\text{max}}
& & \sum_{l\in \mathcal{L}}\min _{\mathbf{r}_l}L_l(\mathbf{r}_l,\mu_l,\boldsymbol{\phi}_l,\boldsymbol{\mu}^{j-1})\label{opt:dual2}\\
& \text{s.t.}
& & \mu_l \geq 0,\boldsymbol{\phi}_l\geq \mathbf{0},\hspace{0.5cm}l\in \mathcal{L},
\end{aligned}
\end{equation} 
iteratively and then update $\alpha_{k,l}^{p,j+1}$ for iteration $j+1$. The above problem is decomposable in $\{\mu_l,\boldsymbol{\phi}_l\}$ and can be solved in parallel for all links. Due to strong duality \cite[p. 226--p. 227]{boyd2004convex}, each subproblem of \eqref{opt:dual2} is equivalent to the following per-link problem in the primal domain:
\begin{equation}
\begin{aligned}
& \underset{\mathbf{r}_l}{\min}
& & \sum_{k=1}^K\sum_{p:\{p\in \mathcal{P}_k,l\in\mathcal{L}_k^p\}}\alpha_{k,l}^{p,j}\psi_k^p(r_k^p)\label{opt:perlink}\\
& \text{s.t.}
& &\sum_{k=1}^K\sum_{p:\{p\in \mathcal{P}_k,l\in\mathcal{L}_k^p\}}r_k^p \leq C_l,  \\
& & &\alpha_{k,l}^{p,j}r_k^p \geq 0,\hspace{0.5cm}  p\in \mathcal{P}_k, \forall k, l\in \mathcal{L}_k^p.
\end{aligned}
\end{equation} 
The optimal $r_k^p$ and $\mu_l$ can be obtained using the first-order optimality condition for the per-link subproblem in \eqref{opt:perlink}. Here, we list KKT conditions as follows:
\allowdisplaybreaks
\begin{subequations}
	\begin{align}
	&\frac{\partial L_l(\mathbf{r}_l,\mu_l,\boldsymbol{\phi}_l,\boldsymbol{\mu}^{j-1})}{\partial r_k^p}=\alpha_{k,l}^{p,j}\frac{\partial \psi_k^p(r_k^p)}{\partial r_k^p}+\mu_l-\alpha_{k,l}^{p,j}\phi_k^p=0,\label{eq:findr}\\
	&\sum_{k=1}^K\sum_{p:\{p\in \mathcal{P}_k,l\in\mathcal{L}_k^p\}}r_k^p \leq C_l, \label{eq:findr11}\\
	& \mu_l\: (\sum_{k=1}^K\sum_{p:\{p\in \mathcal{P}_k,l\in\mathcal{L}_k^p\}}r_k^p - C_l)=0, \:\:\mu_l\geq 0, \label{eq:findr1}
	\\&	 \alpha_{k,l}^{p,j}r_k^p\:\phi_k^p=0, \:\:r_k^p \geq 0,\:\:\phi_k^p\geq 0.\label{eq:findr2}
	\end{align}
\end{subequations}
First, we consider that $r_k^p>0$ and $\phi_k^p=0$. Due to the strict convexity of $\psi_k^p(r_k^p)$, $\partial \psi^p_k(r_k^p)/\partial r^p_k$ is strictly increasing. Thus, given $\mu_l$, there is a unique $r_k^p$ to solve $\alpha_{k,l}^{p,j}\partial \psi_k^p(r_k^p)/\partial r_k^p+\mu_l=0$.
Since $\partial \psi^p_k(r_k^p)/\partial r^p_k$ is strictly increasing, we implement a bisection search on $r_k^p$ in the non-negative orthant $r_k^p\geq 0$ to find $r_k^p$ from $\alpha_{k,l}^{p,j}\partial \psi_k^p(r_k^p)/\partial r_k^p+\mu_l=0$. If the obtained $r_k^p$ is positive, we keep $\phi_k^p=0$. Otherwise, we set $r_k^p=0$ and find $\phi_k^p=\partial \psi_k^p(r_k^p)/\partial r_k^p|_{r_k^p=0}+\mu_l/\alpha_{k,l}^{p,j}$. For a given $\mu_l$, we obtain each $r_k^p$ variable associated with link $l$, i.e., $r_k^p: p\in\mathcal{P}_k,l\in\mathcal{L}_k^p,k=1:K$. The dual approach for solving the optimization in \eqref{opt:perlink} works as follows: implement a bisection search 
on the Lagrange multiplier $\mu_l$ in the positive orthant and numerically find each $r_k^p$ variable from \eqref{eq:findr} and \eqref{eq:findr2} for each $\mu_l$ until we have $\sum_{p:\{p\in \mathcal{P}_k,l\in\mathcal{L}_k^p\}}r_k^p = C_l$. If there is no such positive $\mu_l$,  we drop the first constraint from optimization \eqref{opt:perlink} and solve \eqref{opt:perlink} by setting the gradient  of the cost function to zero. Then, we project the solution to the positive orthant. Due to the strict convexity of each subproblem, the optimal primal variables are unique. The optimal $\mu_l$ for each per-link subproblem is also unique. We justify this claim.
\begin{algorithm}[t!]
	0.	\textbf{Initialization}  $s_1=0$, $s_2=\text{large number}$, $q_1=0$, $q_2=0$, $q_3=0$\;
	\Repeat{$s_2-s_1$ is small enough}{
		1. $s_3=(s_1+s_2)/2$\;
		2. Implement a bisection search to solve \eqref{eq:findr} with $\phi_k^p=0$ and find $r_k^p:r_k^p\geq 0$, where $\mu_l=s_1$\;
		3. \If{there is no positive solution for $r_k^p$}{
			$r_k^p=0$ and $\phi_k^p=\partial \psi_k^p(r_k^p)/\partial r_k^p|_{r_k^p=0}+\mu_l/\alpha_{k,l}^{p,j}$\;
		}
		4. $q_1=\sum_{k=1}^K\sum_{p:\{p\in \mathcal{P}_k,l\in\mathcal{L}_k^p\}}r_k^p - C_l$\;
		5. Implement a bisection search to solve \eqref{eq:findr} with $\phi_k^p=0$ and find $r_k^p:r_k^p\geq 0$, where $\mu_l=s_2$\;
		6. \If{there is no positive solution for $r_k^p$}{
			$r_k^p=0$ and $\phi_k^p=\partial \psi_k^p(r_k^p)/\partial r_k^p|_{r_k^p=0}+\mu_l/\alpha_{k,l}^{p,j}$\;
		}
		7. $q_2=\sum_{k=1}^K\sum_{p:\{p\in \mathcal{P}_k,l\in\mathcal{L}_k^p\}}r_k^p - C_l$\;
		8. Implement a bisection search to solve \eqref{eq:findr} with $\phi_k^p=0$ and find $r_k^p:r_k^p\geq 0$, where $\mu_l=s_3$\;
		9. \If{there is no positive solution for $r_k^p$}{
			$r_k^p=0$ and $\phi_k^p=\partial \psi_k^p(r_k^p)/\partial r_k^p|_{r_k^p=0}+\mu_l/\alpha_{k,l}^{p,j}$\;
		}
		10. $q_3=\sum_{k=1}^K\sum_{p:\{p\in \mathcal{P}_k,l\in\mathcal{L}_k^p\}}r_k^p - C_l$\;
		11. \If{$q_1.q_3 <0$}{
			$s_2=s_3$\;
		}
		12. \If{$q_2.q_3 <0$}{
			$s_1=s_3$\;
		}
		13. \If{$q_1 < 0$, $q_2< 0$, $q_3 <0$ }{
			13.1. $\mu_l=0$\;
			13.2. Solve $\partial \psi_k^p(r_k^p)/\partial r_k^p=0$ to find $r_k^p$\; 
			13.3. Project the obtained $r_k^p$ variable to the positive orthant\;
			13.4. $s_2=s_1$\;
		}
	}     
	\caption{Dual algorithm to solve the per-link optimization in \eqref{opt:perlink}}
	\label{al:bi}
\end{algorithm}
 If the link capacity constraint is not tight, then due to \eqref{eq:findr1}, $\mu_l$ has to be zero. If the link capacity is tight, then at least one $r_k^p: p\in \mathcal{P}_k,l\in\mathcal{L}_k^p,k=1:K$ is non-zero and $\phi_k^p=0$. Due to a) the strict convexity of $\psi^p_k(r_k^p)$ and the monotone variation of $\psi^p_k(r_k^p)/\partial r^p_k$; and b) the uniqueness of the optimal $r_k^p$, the obtained $\mu_l$ from \eqref{eq:findr} is unique. We justify the bisection search on $\mu_l$ as follows: if the unique optimal $\mu_l$ is positive, from \eqref{eq:findr1}, we observe that we must have $\sum_{p:\{p\in \mathcal{P}_k,l\in\mathcal{L}_k^p\}}r_k^p = C_l$, where each $r_k^p$ is found from \eqref{eq:findr} and \eqref{eq:findr2}. Such positive $\mu_l$ can be uniquely found using a bisection search due to the strictly monotone variation of $\sum_{p:\{p\in \mathcal{P}_k,l\in\mathcal{L}_k^p\}}r_k^p,$ with $\mu_l$ (strict convexity of $\psi_k^p(r_k^p)$ as explained above). If the optimal $\mu_l$ is zero, then \eqref{eq:findr11} and \eqref{eq:findr1} are already satisfied and it is enough to find the unique non-negative minimizer of each $\psi_k^p(r_k^p)$ from \eqref{eq:findr} and \eqref{eq:findr2}. In light of the above arguments, two nested bisection methods are required to solve \eqref{opt:perlink}: the inner bisection works on $r_k^p$ and the outer one works on $\mu_l$. 
The summary of the proposed bisection approach to solve the per-link optimization in \eqref{opt:perlink} is given in Algorithm \ref{al:bi}. 

\begin{algorithm}[t!]
	0.	\textbf{Initialization} Assign some small positive number to each $\mu_l^0$, $j=0$; \\
	\Repeat{$\boldsymbol{\mu}^j$ converge}{
		\For{all links}{ 
			1. Find  $\alpha_{k,l}^{p,j+1}=\mu_l^j/\sum_{l'\in \mathcal{L}_k^p}\mu_{l'}^j$\;
			\If{$\mu_l^j >0$}{2. Apply Algorithm \ref{al:bi} to find $\mu_l^{j+1}$\;
			}
			3. $j=j+1$\;
		}
	}     
	\For{all $\{r_k^p\}_{p\in\mathcal{P}_k,k=1:K}$ variables}{4. Use the latest computed $r_k^p$ by Algorithm \ref{al:bi} from a per-link subproblem, where $l\in \mathcal{L}_k^p$ and $\mu_l^j>0$\;}
	\caption{Multi-path routing algorithm to solve the optimization in \eqref{opt:ronly}}
	\label{al:parallel}
\end{algorithm}
Suppose that the optimization in \eqref{opt:perlink} is iteratively solved in parallel for all links of the network. For a link with a large capacity, the link capacity constraint is not tight and Algorithm \ref{al:bi} finds $\mu_l^j=0$ and we have $\alpha_{k,l}^{p,j+1}=0$. For those links, we do not need to continue computation as the KKT conditions listed in \eqref{eq:findr}--\eqref{eq:findr2} remain satisfied. In the following iterations, we ignore those links and consider links with $\mu_l^j > 0$. We alternate between solving the optimization in \eqref{opt:perlink} in parallel for all links and updating $\alpha_{k,l}^{p,j+1}$ until all $\{\mu_l^j\}_{l\in\mathcal{L}}$ variables converge, i.e., $\norm{\boldsymbol{\mu}^{j}-\boldsymbol{\mu}^{j-1}}_2<\epsilon$. Once $\boldsymbol{\mu}^{j}$ converges, for each $r_k^p$ variable, we use the computed $r_k^p$ in the last iteration of Algorithm \ref{al:bi} from a subproblem with $\mu_l^j>0, l\in \mathcal{L}_k^p$. A brief description of the proposed dual algorithm for solving the optimization in \eqref{opt:ronly} is given in Algorithm \ref{al:parallel}.

After Algorithm \ref{al:parallel} converges, we use the obtained $r_k^p$ from a per-link problem with tight link capacity constraint, i.e., $l: \mu_l^j>0$, for the other links on that path for which Algorithm \ref{al:parallel} finds $\mu_l^j=0$. The key property of Algorithm \ref{al:parallel} is that after convergence, the obtained $r_k^p$ on different links of one path are identical.
\begin{proposition}
	Upon convergence of Algorithm \ref{al:parallel}, the flow rates across links on each path are identical. 
\end{proposition} 
\begin{proof}
	Algorithm \ref{al:parallel} finds $r_k^p$, from the per-link subproblem for link $l\in\mathcal{L}_k^p,\mu_l^j>0$, using the following equation:
	\begin{align}
	&\frac{\partial L_l(\mathbf{r}_l,\mu_l^j,\boldsymbol{\phi}_l^j,\boldsymbol{\mu}^{j-1})}{\partial r_k^p}=\alpha_{k,l}^{p,j}\frac{\partial \psi_k^p(r_k^p)}{\partial r_k^p}+\mu_l^j-\alpha_{k,l}^p\phi_k^{p,j}\nonumber\\
	&=\frac{\mu_l^{j-1}}{\sum_{l'\in \mathcal{L}_k^p}\mu_{l'}^{j-1}}\frac{\partial \psi_k^p(r_k^p)}{\partial r_k^p}+\mu_l^j-\frac{\mu_l^{j-1}}{\sum_{l'\in \mathcal{L}_k^p}\mu_{l'}^{j-1}}\phi_k^{p,j}=0.\label{eq:cons1}
	\end{align} 
	Suppose Algorithm \ref{al:parallel} has converged in the $j^{\text{\text{th}}}$ iteration; we have  $\norm{\boldsymbol{\mu}^{j}-\boldsymbol{\mu}^{j-1}}_2<\epsilon$. Then, we have $\boldsymbol{\mu}^{j}-\boldsymbol{\mu}^{j-1}=\boldsymbol{\vartheta}$, where $\norm{\boldsymbol{\vartheta}}_2<\epsilon$. We have
	\begin{align}
	\frac{1}{\alpha_{k,l}^{p,j}}=\frac{\sum_{l'\in \mathcal{L}_k^p}\mu_{l'}^{j-1}}{\mu_l^{j-1}}=\frac{\sum_{l'\in \mathcal{L}_k^p}(\mu_{l'}^{j}-\vartheta_{l'})}{\mu_l^{j}-\vartheta_l}.\nonumber
	\end{align}
	We multiply \eqref{eq:cons1} by $1/\alpha_{k,l}^{p,j}$ and we have the following:
	\begin{align}
	&\frac{\partial \psi_k^p(r_k^p)}{\partial r_k^p}+\frac{\sum_{l'\in \mathcal{L}_k^p}(\mu_{l'}^{j}-\vartheta_{l'})}{\mu_l^{j}-\vartheta_l}\mu_l^j-\phi_k^{p,j}\nonumber\\
	&=\frac{\partial \psi_k^p(r_k^p)}{\partial r_k^p}+\sum_{l'\in \mathcal{L}_k^p}(\mu_{l'}^{j}-\vartheta_{l'})(1+\frac{\vartheta_l}{\mu_l^{j}-\vartheta_l})-\phi_k^{p,j}\nonumber\\
	&=\underbrace{\frac{\partial \psi_k^p(r_k^p)}{\partial r_k^p}+\sum_{l\in \mathcal{L}_k^p}\mu_{l}^{j}-\phi_k^{p,j}}_{\partial L_c(\mathbf{r},\boldsymbol{\mu}^{j},\boldsymbol{\phi}^{j})/\partial r_k^p}-\sum_{l'\in \mathcal{L}_k^p}\vartheta_{l'}(1+\frac{\vartheta_l}{\mu_l^{j}-\vartheta_l})\nonumber\\
	&+\frac{\vartheta_l(\sum_{l'\in \mathcal{L}_k^p}\mu_{l'}^j)}{\mu_l^{j}-\vartheta_l}=\frac{\partial L_c(\mathbf{r},\boldsymbol{\mu}^{j},\boldsymbol{\phi}^{j})}{\partial r_k^p}-\sum_{l'\in \mathcal{L}_k^p}\vartheta_{l'}(1+\frac{\vartheta_l}{\mu_l^{j}-\vartheta_l})\nonumber\\
	&+\frac{\vartheta_l(\sum_{l'\in \mathcal{L}_k^p}\mu_{l'}^j)}{\mu_{l}^{j}-\vartheta_l}=0.\label{eq:cons2}
	\end{align}
	When $\epsilon$ tends to zero, then $\boldsymbol{\vartheta}\rightarrow \mathbf{0}$ and from \eqref{eq:cons2} we find $\partial L_c(\mathbf{r},\boldsymbol{\mu}^j,\boldsymbol{\phi}^j)/\partial r_k^p=0$. 
	Moreover, we observe that $\partial L_c(\mathbf{r},\boldsymbol{\mu}^j,\boldsymbol{\phi}^j)/\partial r_k^p$ is independent of the link index on path $p$. This means that $\{r_k^p\}_{p\in\mathcal{P}_k}$ variables obtained
	by solving the link subproblems are identical for all links along each path $p$ for which $\mu_l^j>0$. They are also equal to the minimizer of Lagrangian function $L_c(\mathbf{r},\boldsymbol{\mu}^j,\boldsymbol{\phi}^j)$ in \eqref{eq:lag}.
\end{proof}

\begin{theorem}\label{thr:1}
	If $\psi(\mathbf{r})$ is strictly convex and separable, then the primal and dual iterates of Algorithm \ref{al:parallel} will converge to the optimal primal and dual solutions of \eqref{opt:ronly}.
\end{theorem}
\begin{proof}
	Notice that based on the definition of $\alpha_{k,l}^{p,j}$, given identical feasible variables $r_k^p$, $\hat{\mu}_l$ and $\hat{\boldsymbol{\phi}_l}$  to both Lagrangian functions in \eqref{eq:lag} and \eqref{eq:decomlag}, from \eqref{eq:equality}, we have $\sum_{l\in \mathcal{L}} L_l(\mathbf{r}_l,\hat{\mu}_l,\hat{\boldsymbol{\phi}}_l,\boldsymbol{\mu}^{j-1})=L_c(\mathbf{r},\hat{\boldsymbol{\mu}},\hat{\boldsymbol{\phi}})$. First, we show that $\sum_{l\in \mathcal{L}}\min _{\mathbf{r}_l}L_l(\mathbf{r}_l,\hat{\mu}_l,\hat{\boldsymbol{\phi}}_l,\boldsymbol{\mu}^{j-1})$ is a lower-bound for $\min _{\mathbf{r}}L_c(\mathbf{r},\hat{\boldsymbol{\mu}},\hat{\boldsymbol{\phi}})$.  Since the minimum of $L_l(\mathbf{r}_l,\hat{\mu}_l,\hat{\boldsymbol{\phi}}_l,\boldsymbol{\mu}^{j-1})$ is less than or equal to the other values of $L_l(\mathbf{r}_l,\hat{\mu}_l,\hat{\boldsymbol{\phi}}_l,\boldsymbol{\mu}^{j-1})$, we have 
		\begin{align}
		\min _{\mathbf{r}_l}L_l(\mathbf{r}_l,\hat{\mu}_l,\hat{\boldsymbol{\phi}}_l,\boldsymbol{\mu}^{j-1})\leq L_l(\mathbf{r}_l,\hat{\mu}_l,\hat{\boldsymbol{\phi}}_l,\boldsymbol{\mu}^{j-1}).\nonumber
		\end{align}
		Thus, we obtain
		 \begin{align}
		&\sum_{l\in \mathcal{L}}\min _{\mathbf{r}_l}L_l(\mathbf{r}_l,\hat{\mu}_l,\hat{\boldsymbol{\phi}}_l,\boldsymbol{\mu}^{j-1})\nonumber\\
		&\leq\sum_{l\in \mathcal{L}} L_l(\mathbf{r}_l,\hat{\mu}_l,\hat{\boldsymbol{\phi}}_l,\boldsymbol{\mu}^{j-1})=L_c(\mathbf{r},\hat{\boldsymbol{\mu}},\hat{\boldsymbol{\phi}}),\nonumber
		\end{align} 
		where the  equality is due to \eqref{eq:equality}. In $L_c(\mathbf{r},\hat{\boldsymbol{\mu}},\hat{\boldsymbol{\phi}})$, we choose $\mathbf{r}$ to be the minimizer of $L_c(\mathbf{r},\hat{\boldsymbol{\mu}},\hat{\boldsymbol{\phi}})$. Thus, we obtain 
		\begin{align}
		&\sum_{l\in \mathcal{L}}\min_{\mathbf{r}_l}L_l(\mathbf{r}_l,\hat{\mu}_l,\hat{\boldsymbol{\phi}}_l,\boldsymbol{\mu}^{j-1})\leq \min _{\mathbf{r}}L_c(\mathbf{r},\hat{\boldsymbol{\mu}},\hat{\boldsymbol{\phi}}).\label{eq:uper}
		\end{align}
		From \eqref{eq:uper}, we observe that solving the problem in \eqref{opt:dual2} iteratively is  a successive lower-bound maximization (upper-bound minimization if we rewrite problems \eqref{opt:dual1} and \eqref{opt:dual2} as minimizations).

	We justify the claim that the primal and dual solutions obtained from solving \eqref{opt:dual2} successively converge to the primal and dual solutions of \eqref{opt:ronly}. We build our proof based on the convergence theory for BSUM given in \cite{razaviyayn2013unified}. 
		We show, in the same order given in the Appendix, that the lower-bound satisfies all four convergence conditions given in \cite[Assumption 2]{razaviyayn2013unified}:
	\begin{enumerate}
		\item At feasible points $\hat{\boldsymbol{\mu}}\geq \mathbf{0}$ and $\hat{\boldsymbol{\phi}}\geq \mathbf{0}$, we show that $\min _{\mathbf{r}}L_c(\mathbf{r},\hat{\boldsymbol{\mu}},\hat{\boldsymbol{\phi}})=\sum_{l\in \mathcal{L}}\min _{\mathbf{r}_l}L_l(\mathbf{r}_l,\hat{\mu}_l,\hat{\boldsymbol{\phi}}_l,\hat{\boldsymbol{\mu}})$. 
			From KKT conditions for each subproblem, we obtain
			\begin{align}
			&\frac{L_l(\mathbf{r}_l,\hat{\mu}_l,\hat{\boldsymbol{\phi}}_l,\hat{\boldsymbol{\mu}})}{\partial r_k^p}=\frac{\hat{\mu}_l}{\sum_{l'\in \mathcal{L}_k^p}\hat{\mu}_{l'}}\frac{\partial \psi_k^p(r_k^p)}{\partial r_k^p}+\hat{\mu}_l\nonumber\\
			&-\frac{\hat{\mu}_l}{\sum_{l'\in \mathcal{L}_k^p}\hat{\mu}_{l'}}\hat{\phi}_k^{p}=0.\label{eq:mini1}
			\end{align}
			Assuming $\hat{\mu}_{l}>0$, after multiplication by $\frac{\sum_{l'\in \mathcal{L}_k^p}\hat{\mu}_{l'}}{\hat{\mu}_l}$, we obtain
			\begin{align}
			\frac{\partial \psi_k^p(r_k^p)}{\partial r_k^p}+\sum_{l\in \mathcal{L}_k^p}\hat{\mu}_l-\hat{\phi}_k^{p}=\frac{\partial L_c(\mathbf{r},\hat{\boldsymbol{\mu}},\hat{\boldsymbol{\phi}})}{\partial r_k^p}=0.\label{eq:mini2}
			\end{align}
			When $\psi(\mathbf{r})$ is strictly convex, there is a unique minimizer for each $L_c(\mathbf{r},\hat{\boldsymbol{\mu}},\hat{\boldsymbol{\phi}})$ and $L_l(\mathbf{r}_l,\hat{\mu}_l,\hat{\boldsymbol{\phi}}_l,\hat{\boldsymbol{\mu}})$. We observe from \eqref{eq:mini1} and \eqref{eq:mini2} that, at feasible points $\hat{\boldsymbol{\mu}}\geq \mathbf{0}$ and $\hat{\boldsymbol{\phi}}\geq \mathbf{0}$, the minimizer of $L_c(\mathbf{r},\hat{\boldsymbol{\mu}},\hat{\boldsymbol{\phi}})$ is equal to that of $L_l(\mathbf{r}_l,\hat{\mu}_l,\hat{\boldsymbol{\phi}}_l,\hat{\boldsymbol{\mu}})$. From \eqref{eq:equality}, applying identical variables $r_k^p$, $\hat{\mu}_l$ and $\hat{\phi}_l$ to both $L_c(\mathbf{r},\hat{\boldsymbol{\mu}},\hat{\boldsymbol{\phi}})$ and $L_l(\mathbf{r}_l,\hat{\mu}_l,\hat{\boldsymbol{\phi}}_l,\hat{\boldsymbol{\mu}})$, we have $L_c(\mathbf{r},\hat{\boldsymbol{\mu}},\hat{\boldsymbol{\phi}},\hat{\boldsymbol{\mu}})=\sum_{l\in \mathcal{L}}L_l(\mathbf{r}_l,\hat{\mu}_l,\hat{\boldsymbol{\phi}}_l,\hat{\boldsymbol{\mu}})$. We choose  each $r_k^p$ to be the minimizer, and we find $\min _{\mathbf{r}}L_c(\mathbf{r},\hat{\boldsymbol{\mu}},\hat{\boldsymbol{\phi}},\hat{\boldsymbol{\mu}})=\sum_{l\in \mathcal{L}}\min _{\mathbf{r}_l}L_l(\mathbf{r}_l,\hat{\mu}_l,\hat{\boldsymbol{\phi}}_l,\hat{\boldsymbol{\mu}})$.
		
		\item From \eqref{eq:uper}, we observe that $\sum_{l\in \mathcal{L}}\min _{\mathbf{r}_l}L_l(\mathbf{r}_l,\hat{\mu}_l,\hat{\boldsymbol{\phi}}_l,\boldsymbol{\mu}^{j-1})$ is a lower-bound.
		
		\item We deploy \cite[Proposition 7.1.1]{bertsekas1999nonlinear} to find the derivative of $\min _{\mathbf{r}_l}L_l(\mathbf{r}_l,\hat{\mu}_l,\hat{\boldsymbol{\phi}}_l,\hat{\boldsymbol{\mu}})$ with respect to $\mu_l$. There are three satisfied conditions that ensure the existence of the derivative: a) the feasible set of \eqref{opt:perlink} is compact; b) $L_l(\mathbf{r}_l,\hat{\mu}_l,\hat{\boldsymbol{\phi}}_l,\hat{\boldsymbol{\mu}})$ is continuous in $\mu_l$; and c)
			for each $\hat{\mu}_l$, the equation $\partial L_l(\mathbf{r}_l,\hat{\mu}_l,\hat{\boldsymbol{\phi}}_l,\hat{\boldsymbol{\mu}})/\partial r_k^p=0$ has a unique solution for $r_k^p$ due to the strict convexity of $\psi(\mathbf{r})$. Given identical $\hat{\mu}_l$ and $\hat{\phi}_l$ to both Lagrangian functions \eqref{eq:lag} and \eqref{eq:decomlag}, the derivative of $\sum_{l\in \mathcal{L}}\min _{\mathbf{r}_l}L_l(\mathbf{r}_l,\hat{\mu}_l,\hat{\boldsymbol{\phi}}_l,\hat{\boldsymbol{\mu}})$ with respect to $\mu_l$ is
			\begin{align}
			&\sum_{k=1}^{K}\sum_{p:\{p\in \mathcal{P}_k,l\in\mathcal{L}_k^p\}}\hspace{-.55cm}r_k^p - C_l,\nonumber\\
			& \hspace{2.1cm}\text{where}\hspace{.1cm} r_k^p=\arg\min_{r_k^p}L_l(\mathbf{r}_l,\hat{\mu}_l,\hat{\boldsymbol{\phi}}_l,\hat{\boldsymbol{\mu}}).\nonumber
			\end{align} 
			The derivative of $\min _{\mathbf{r}}L_c(\mathbf{r},\hat{\boldsymbol{\mu}},\hat{\boldsymbol{\phi}})$ with respect to $\mu_l$ is
			\begin{align}
			\sum_{k=1}^{K}\sum_{p:\{p\in \mathcal{P}_k,l\in\mathcal{L}_k^p\}}\hspace{-.65cm}r_k^p - C_l, \hspace{.1cm}\text{where}\hspace{.1cm} r_k^p=\arg\min_{r_k^p}L_c(\mathbf{r},\hat{\boldsymbol{\mu}},\hat{\boldsymbol{\phi}}).\nonumber
			\end{align}  As the minimizers of $L_c(\mathbf{r},\hat{\boldsymbol{\mu}},\hat{\boldsymbol{\phi}})$ and $L_l(\mathbf{r}_l,\hat{\mu}_l,\hat{\boldsymbol{\phi}}_l,\hat{\boldsymbol{\mu}})$ are equal at point $\hat{\boldsymbol{\mu}}$ due to 	\eqref{eq:mini1} and \eqref{eq:mini2}, we observe that both above derivatives are equal.
		
		\item $\sum_{l\in \mathcal{L}}\min _{\mathbf{r}_l}L_l(\mathbf{r}_l,\hat{\mu}_l,\hat{\boldsymbol{\phi}}_l,\boldsymbol{\mu}^{j-1})$ is a piecewise linear function of $\hat{\boldsymbol{\mu}}$, and thus, it is a  continuous function of $\hat{\boldsymbol{\mu}}$. 
	\end{enumerate}
\begin{algorithm}
	0.	\textbf{Initialization} Choose a feasible vector $\mathbf{r}^0$, $m=0$\;
	\Repeat{variables in $\mathbf{r}^m$ converge}{
		1. Find  the upper-bound \eqref{eq:lip} using $\mathbf{r}^m$\;
		2. Apply Algorithm \ref{al:parallel}  to find $\mathbf{r}$\;
		3. $m=m+1$ and $\mathbf{r}^m=\mathbf{r}$\;
	}     
	\caption{Multi-path routing algorithm for non-separable cost functions}
	\label{al:upper}
\end{algorithm}
	Building on the above arguments, Algorithm \ref{al:parallel} is a block successive lower-bound maximization method, which satisfies all four convergence conditions given in \cite[Assumption 2]{razaviyayn2013unified}. Algorithm \ref{al:parallel} converges to the global optimal solution of the concave problem \eqref{opt:dual1} \cite[Theorem 2]{razaviyayn2013unified}, which has an identical objective function to \eqref{opt:ronly} at the optimal point as a result of strong duality \cite[p. 226--p. 227]{boyd2004convex}. Once Algorithm \ref{al:parallel} converges, $\boldsymbol{\mu}^j$ and $r_k^p=\arg\min_{r_k^p} L_l(\mathbf{r}_l,\mu_l^j,\boldsymbol{\phi}_l^j,\boldsymbol{\mu}^{j-1})$ by Algorithm \ref{al:parallel} satisfy \eqref{eq:findr}--\eqref{eq:findr2}. The KKT conditions for \eqref{opt:ronly} are \eqref{eq:findr11}--\eqref{eq:findr2} in addition to $\frac{\partial L_c(\mathbf{r},\boldsymbol{\mu}^j,\boldsymbol{\phi}^j)}{\partial r_k^p}=0$. Due to \eqref{eq:cons1} and \eqref{eq:cons2}, when Algorithm \ref{al:parallel} converges, minimizers of $L_l(\mathbf{r}_l,\mu_l^j,\boldsymbol{\phi}_l^j,\boldsymbol{\mu}^{j-1})$ and $L_c(\mathbf{r},\boldsymbol{\mu}^j,\boldsymbol{\phi}^j)$ are identical, and thus, \eqref{eq:findr} ensures $\frac{\partial L_c(\mathbf{r},\boldsymbol{\mu}^j,\boldsymbol{\phi}^j)}{\partial r_k^p}=0$.	Hence, the primal and dual variables obtained by Algorithm \ref{al:parallel} satisfy the KKT conditions for \eqref{opt:ronly}. 
\end{proof}

\begin{remark}\label{rem:nonsep}
	If the cost function is convex and has a gradient that is Lipschitz continuous, but the function is not separable in $r_k^p$, i.e., $\psi(\mathbf{r})$ cannot be written as $\psi(\mathbf{r})=\sum_{k=1}^{K}\sum_{p\in \mathcal{P}_k}\psi_k^p(r_k^p)$, we use the quadratic upper-bound given in \cite[eq. (12)]{7366709}, which is separable in variables. For an arbitrary convex cost function with a Lipschitz continuous gradient like $\psi(\mathbf{r})$, we have the following upper-bound:
	\begin{align}
	\psi(\mathbf{\mathbf{r}})\leq \psi(\mathbf{r}^m)+\nabla \psi(\mathbf{r}^m)^T(\mathbf{r}-\mathbf{r}^m)+\frac{\gamma}{2}\norm{\mathbf{r}-\mathbf{r}^m}_2^2,\label{eq:lip}
	\end{align}
	where $\gamma$ is the Lipschitz constant, and $\mathbf{r}^m=\{r_k^{p,m}\}_{p\in\mathcal{P}_k,k=1:K}$ is the $m^{\text{th}}$ iterate in the successive upper-bound minimization. We start from an initial point $\mathbf{r}^0$ in the feasible set and find the upper-bound \eqref{eq:lip}. Then, we apply Algorithm \ref{al:parallel} to solve the problem with the upper-bound \eqref{eq:lip} to the global optimal solution in a parallel fashion. 
	
	When the upper-bound \eqref{eq:lip} is substituted for the cost function, the first KKT condition is
	\begin{align}
	&\frac{\partial L_l(\mathbf{r}_l,\mu_l,\boldsymbol{\phi}_l,\boldsymbol{\mu}^{j-1})}{\partial r_k^p}=\nonumber\\
	&\alpha_{k,l}^{p,j}\left(\frac{\partial \psi_k^p(r_k^p)}{\partial r_k^p}\mid_{r_k^P=r_k^{p,m}}+\gamma(r_k^p-r_k^{p,m})\right)+\mu_l-\alpha_{k,l}^{p,j}\phi_k^{p}=0,
	\end{align}instead of \eqref{eq:findr}.
	Once the problem with the upper-bound \eqref{eq:lip} is solved, we use the obtained solution to update $\mathbf{r}^m$ in the upper-bound \eqref{eq:lip}. We repeat this approach until $\mathbf{r}^m$ converges. We summarize this approach in Algorithm \ref{al:upper}.
	
	In iteration $m$, the value of the upper-bound \eqref{eq:lip} and its gradient are $\psi(\mathbf{r}^m)$ and $\nabla \psi(\mathbf{r}^m)$, respectively, which are equal to the value and the gradient of the non-separable cost function $\psi(\mathbf{r})$. Furthermore, the upper-bound in \eqref{eq:lip} is continuous, and thus, all four convergence conditions given in \cite[Assumption 2]{razaviyayn2013unified} and listed in the Appendix are satisfied. Due to \cite[Theorem 2]{razaviyayn2013unified} and the convexity of the non-separable cost function, the obtained solution by Algorithm \ref{al:upper}, which implements BSUM, is identical to the solution of the original problem with the non-separable cost function.
\end{remark}

\begin{remark}\label{co:2}
	When the cost function is convex and separable, but not strictly convex, we add a proximal term to the cost function and make it locally strongly convex as follows:
	\allowdisplaybreaks
	\begin{equation}
	\begin{aligned}
	& \underset{\mathbf{r}}{\min}
	& & \sum_{k=1}^{K}\sum_{p\in \mathcal{P}_k}\psi_k^p(r_k^p)+\frac{\kappa}{2}\norm{\mathbf{r}-\mathbf{r}^m}_2^2\label{opt:convex}\\
	& \text{s.t.}
	& & \eqref{eq:linkcap}, r_k \geq 0,\hspace{0.5cm} p\in \mathcal{P}_k,\forall k,
	\end{aligned}
	\end{equation} 
	where $\kappa$ is a small positive constant. We use Algorithm \ref{al:parallel} to solve the above problem when we use the following equation:
	\begin{align}
	&\frac{\partial L_l(\mathbf{r}_l,\mu_l,\boldsymbol{\phi}_l,\boldsymbol{\mu}^{j-1})}{\partial r_k^p}=\alpha_{k,l}^{p,j}\frac{\partial \psi_k^p(r_k^p)}{\partial r_k^p}+\alpha_{k,l}^{p,j}\kappa(r_k^p-r_k^{p,m})\nonumber\\
	&+\mu_l-\alpha_{k,l}^{p,j}\phi_k^{p}=0,\nonumber
	\end{align}
	instead of \eqref{eq:findr} to find $r_k^p$, where $r_k^{p,m}$ is the value of $r_k^p$ in the $m^{\text{\text{th}}}$ iteration of solving \eqref{opt:convex}. We successively solve \eqref{opt:convex} with Algorithm \ref{al:parallel} and update $\mathbf{r}^m$ until $\mathbf{r}^m$ converges. Similar to Remark 2, one can show that the cost function with the proximal term in \eqref{opt:convex} satisfies the four convergence conditions in \cite[Assumption 2]{razaviyayn2013unified} and the global minimum is obtained after successive minimizations, since each local minimum is also global for a convex function.
\end{remark}

\section{Simultaneous Resource Reservations in the Backhaul and RAN}
In this section, we study the joint link capacity and AP transmission resource reservation based on the user demand and downlink statistics. Prior to the observation of user demands, based on the formulated model in \eqref{opt:first}, the network operator finds the optimal amount of reserved resources in the backhaul and APs such that neither the link capacity nor AP capacity is exceeded.

\subsection{Resource Reservation in the Backhaul}
Let us drop the equality constraint \eqref{eq:split} from \eqref{opt:first} and substitute $\sum_{p \in \mathcal{P}_k} r_k^p$ for $r_k$. Then, we have
\begin{equation}
\begin{aligned}
& \underset{\mathbf{r,t}}{\min}
& & \sum_{k=1}^{K}\Big[-\mathbb{E}[\min(\sum_{p \in \mathcal{P}_k} r_k^p,d_k)]\\
&&&\hspace{1cm}+\theta_k\sum_{p \in \mathcal{P}_k}\int_{0}^{r_k^p} z_k^p(v_k^p,t_{k}^{p})\:(r_k^p-v_k^p)dv_k^p\Big]\label{eq:sim}\\
& \text{s.t.}
& & \eqref{eq:linkcap},\eqref{eq:nodecap}, r_k^p, t_{k}^{p} \geq 0,\hspace{0.5cm} p\in \mathcal{P}_k,\forall k.
\end{aligned}
\end{equation}
We solve the above problem using the proposed BCD algorithm. With the fixed $\mathbf{t}$, we minimize \eqref{eq:sim} with respect to $\mathbf{r}$ and update it. With updated $\mathbf{r}$, we minimize \eqref{eq:sim} with respect to $\mathbf{t}$ and update it. We underline the iterates of the BCD algorithm. In the $i+1^{\text{th}}$ iteration of the BCD algorithm, fixing $\underline{\mathbf{t}}^{i}$, we minimize with respect to $\mathbf{r}$. Then, the minimization problem in \eqref{eq:sim} reduces to the following \textit{convex} one:  
\begin{equation}
\begin{aligned}
& \underset{\mathbf{r}}{\min}
& & \sum_{k=1}^{K}\Big[-\mathbb{E}[\min(\sum_{p \in \mathcal{P}_k} r_k^p,d_k)]\nonumber\\
&&&\hspace{1cm}+\theta_k\sum_{p \in \mathcal{P}_k}\int_{0}^{r_k^p} z_k^p(v_k^p,\underline{t}_{k}^{p,i})\:(r_k^p-v_k^p)dv_k^p\Big]\\
& \text{s.t.}
& & \eqref{eq:linkcap}, r_k^p \geq 0,\hspace{0.5cm} p\in \mathcal{P}_k,\forall k.\label{opt:app}
\end{aligned}
\end{equation}
It is observed that although the expected outage is separable in $r_k^p$, variables are coupled in the first term of the objective function. Therefore, we substitute the global quadratic upper-bound given in \eqref{eq:lip} for the expected supportable traffic demand. First, let us calculate the Lipschitz constant for the gradient of  $-\mathbb{E}[\min(\sum_{p \in \mathcal{P}_k} r_k^p,d_k)]$. 
The second derivative of $-\mathbb{E}[\min(\sum_{p \in \mathcal{P}_k} r_k^p,d_k)]$ is
\begin{align}
-\frac{\partial^2 \mathbb{E}[\min(\sum_{p \in \mathcal{P}_k} r_k^p,d_k)]}{\partial r_k^p \partial r_k^{p'}}=f_k(\sum_{p \in \mathcal{P}_k} r_k^p).\nonumber
\end{align}
The Hessian matrix for $-\mathbb{E}[\min(\sum_{p \in \mathcal{P}_k} r_k^p,d_k)]$ is $|\mathcal{P}_k|\times|\mathcal{P}_k|$ dimensional, where all entries are $f_k(\sum_{p \in \mathcal{P}_k} r_k^p)$.
The eigenvalues of the Hessian matrix are all zeros except one of them, which is $|\mathcal{P}_k|f_k(\sum_{p \in \mathcal{P}_k} r_k^p)$. Therefore, the Lipschitz constant is $|\mathcal{P}_k|$. We now place the Lipschitz constant in \eqref{eq:lip} and find the upper-bound which is separable in $r_k^p$  as follows:
\begin{align}
&-\mathbb{E}[\min(\sum_{p \in \mathcal{P}_k} r_k^p,d_k)]\leq -\mathbb{E}[\min(\sum_{p \in \mathcal{P}_k}r_k^{p,m},d_k)]\nonumber\\&+(F_k(\sum_{p \in \mathcal{P}_k} r_k^{p,m})-1)(\sum_{p \in \mathcal{P}_k} r_k^p-\sum_{p \in \mathcal{P}_k} r_k^{p,m})\nonumber\\&+\frac{|\mathcal{P}_k|}{2}\sum_{p \in \mathcal{P}_k}(r_k^{p,m}- r_k^p)^2,\label{eq:updemand}
\end{align}
where $r_k^{p,m}$ is the $m^{\text{\text{th}}}$ iterate. We substitute upper-bound \eqref{eq:updemand} for the expected supportable demand and the optimization problem in each iteration becomes:
\begin{align}
& \underset{\mathbf{r}}{\min}
& & \sum_{k=1}^{K}\Bigg[(F_k(\sum_{p \in \mathcal{P}_k} r_k^{p,m})-1)(\sum_{p \in \mathcal{P}_k} r_k^p-\sum_{p \in \mathcal{P}_k} r_k^{p,m})\;\label{opt:upp}\\& & &+\frac{|\mathcal{P}_k|}{2}\sum_{p \in \mathcal{P}_k}(r_k^p-r_k^{p,m})^2\nonumber\\& & &+\theta_k\sum_{p \in \mathcal{P}_k}\int_{0}^{r_k^p} \hspace{-.1cm}z_k^p(v_k^p,\underline{t}_{k}^{p,i})\:(r_k^p-v_k^p)dv_k^p\Bigg]\nonumber\\
& \text{s.t.}
& & \eqref{eq:linkcap}, r_k^p \geq 0,\hspace{0.5cm} p\in \mathcal{P}_k,\forall k.\nonumber
\end{align}
We leverage Algorithm \ref{al:upper} to solve \eqref{opt:app} in a parallel fashion. In each iteration of Algorithm \ref{al:upper}, Algorithm \ref{al:parallel} is called to solve the problem in \eqref{opt:upp}. Moreover, Algorithm \ref{al:bi} is called within Algorithm \ref{al:parallel} and it needs to solve $\frac{\partial L_l(\mathbf{r}_l,\mu_l,\boldsymbol{\phi}_l,\boldsymbol{\mu}^{j-1})}{\partial r_k^p}=0$. We rewrite \eqref{eq:findr} for the above optimization problem in the $j^{\text{th}}$ iteration of Algorithm \ref{al:parallel} as follows:
\begin{align}
&\frac{\partial L_l(\mathbf{r}_l,\mu_l,\boldsymbol{\phi}_l,\boldsymbol{\mu}^{j-1})}{\partial r_k^p}=\alpha_{k,l}^{p,j}(F_k(\sum_{p \in \mathcal{P}_k} r_k^{p,m})-1)+\mu_l\nonumber\\
&+\alpha_{k,l}^{p,j}|\mathcal{P}_k|( r_k^p- r_k^{p,m})+\theta_k\:\alpha_{k,l}^{p,j} Z_k^p(r_k^p,\underline{t}_k^{p,i})-\alpha_{k,l}^{p,j}\phi_k^p=0.\nonumber
\end{align}
We observe that for each $\mu_l$, we are able to obtain $r_k^p$ numerically using  $r_k^{p,m}$, independent of the other variables. The solution obtained by Algorithm \ref{al:upper} is unique due to the strong convexity of \eqref{opt:upp} and is global minima as explained in Remark \ref{rem:nonsep}. After Algorithm \ref{al:upper} converges, we set $\underline{\mathbf{r}}^{i+1}=\mathbf{r}^m$.
\begin{remark}
	Suppose that the number of paths that are available to user $k$ and share downlink $w\in \mathcal{W}_k$ is $\varphi_k^w$.	When multiple paths for serving a user share one downlink, we substitute the quadratic upper-bound \eqref{eq:lip} for the outage \eqref{eq:mul} as follows:
	\begin{align}
	&\eqref{eq:mul}\leq	
	\sum_{w \in \mathcal{W}_k}\int_{0}^{\sum_{p:\{p\in \mathcal{P}_k,w\in p\}}r_k^{p,m}} \hspace{-0.8cm}z_k^w(v_k^p,\underline{t}_{k}^{w,i})(\hspace{-.5cm}\sum_{p:\{p\in \mathcal{P}_k,w\in p\}}\hspace{-.5cm}r_k^{p,m}-v_k^p)dv_k^p\nonumber\\
	&+\sum_{w \in \mathcal{W}_k}Z_k^w(\hspace{-.1cm}\sum_{p:\{p\in \mathcal{P}_k,w\in p\}}\hspace{-.1cm}r_k^{p,m},\underline{t}_k^{w,i})(\hspace{-.1cm}\sum_{p:\{p\in \mathcal{P}_k,w\in p\}}\hspace{-.1cm}(r_k^p-r_k^{p,m}))\nonumber\\
	&+\sum_{w \in \mathcal{W}_k}\sum_{p:\{p\in \mathcal{P}_k,w\in p\}}\frac{\varphi_k^w}{2}( r_k^{p}- r_k^{p,m})^2,\label{eq:uppermulti}
	\end{align}
	where $\varphi_k^w$ is the Lipschitz constant. In this case, the objective function of \eqref{opt:upp} is obtained from adding \eqref{eq:updemand} and the RHS of \eqref{eq:uppermulti}. We rewrite \eqref{eq:findr} in the $j^{\text{th}}$ iteration of Algorithm \ref{al:parallel} for this case as follows:
	\begin{align}
	&\frac{\partial L_l(\mathbf{r}_l,\mu_l,\boldsymbol{\phi}_l,\boldsymbol{\mu}^{j-1})}{\partial r_k^p}=\alpha_{k,l}^{p,j}(F_k(\sum_{p \in \mathcal{P}_k} r_k^{p,m})-1)\nonumber\\
	&+\alpha_{k,l}^{p,j}|\mathcal{P}_k|( r_k^p- r_k^{p,m})+\theta_k\:\alpha_{k,l}^{p,j} Z_k^w(\sum_{p:\{p\in \mathcal{P}_k,w \in p\}}r_k^{p,m},\underline{t}_{k}^{w,i})\nonumber\\
	&+\theta_k\varphi_k^w\alpha_{k,l}^{p,j}( r_k^{p}- r_k^{p,m})+\mu_l-\alpha_{k,l}^{p,j}\phi_k^p=0.\nonumber
	\end{align}
\end{remark}

\subsection{Resource Reservation in RAN}
When we minimize \eqref{eq:sim} with respect to $\mathbf{t}$ in the BCD algorithm, we use $\underline{\mathbf{r}}^{i+1}$ obtained by Algorithm \ref{al:upper}. We propose a dual approach to minimize with respect to $\mathbf{t}$. The objective function of \eqref{eq:sim} is separable in $t_k^p$. We are able to parallelize the algorithm across APs since each AP has a separate transmission resource capacity constraint. However, the problem in \eqref{eq:sim} is not necessarily convex in $t_k^p$ for an arbitrary $z_k^p(v_k^p,t_k^p)$. To tackle the potential non-convexity of the problem, we use the BSUM method and convexify the problem locally. We iteratively solve a sequence of convex approximations. Suppose that for each outage term in the objective function of \eqref{eq:sim}, we add a proximal term $\frac{\zeta_k^{p,j}}{2}\norm{t_k^p-t_k^{p,j}}_2^2$, $\zeta_k^{p,j} >0$, to make it locally strongly convex. In the proximal term, $t_k^{p,j}$ is the value of $t_k^{p}$ in the $j^{\text{\text{th}}}$ iteration of successively minimizing \eqref{eq:sim} with respect to $\mathbf{t}$. The objective function with the proximal terms
is an upper-bound of the original objective function. We find the Lagrangian for \eqref{eq:sim} with respect to $\mathbf{t}$ with proximal terms in the objective function as follows:
\begin{align}
&L_t(\mathbf{t},\boldsymbol{\lambda},\boldsymbol{\beta})=\sum_{k=1}^{K}\sum_{p \in \mathcal{P}_k}\Bigg(\theta_k\int_{0}^{\underline{r}_k^{p,i+1}}\hspace{-.3cm} z_k^p(v_k^p,t_{k}^{p})\:(\underline{r}_k^{p,i+1}-v_k^p)dv_k^p\nonumber\\
&+\frac{\theta_k\zeta_k^{p,j}}{2}\norm{t_k^p-t_k^{p,j}}_2^2\Bigg)+\sum_{b\in \mathcal{B}}\lambda_b(\sum_{k=1}^K\sum_{p:\{p\in \mathcal{P}_k,b \in \mathcal{U}_k^p\}}\hspace{-.4cm}t_k^p - C_b)\nonumber\\
&-\sum_{k=1}^K\sum_{p\in \mathcal{P}_k}\beta_k^p t_k^p,\nonumber
\end{align}
where $\underline{\mathbf{r}}^{i+1}$ block is kept fixed. We can decompose the above Lagrangian across APs as follows:
\begin{align}
&L_{t,b}(\mathbf{t}_b,\lambda_b,\boldsymbol{\beta}_b)=\nonumber\\
&\theta_k\sum_{k=1}^K\sum_{p:\{p\in \mathcal{P}_k,b \in \mathcal{U}_k^p\}}\Bigg(\int_{0}^{\underline{r}_k^{p,i+1}} z_k^p(v_k^p,t_{k}^{p})\:(\underline{r}_k^{p,i+1}-v_k^p)dv_k^p\nonumber\\
&+\frac{\zeta_k^{p,j}}{2}\norm{t_k^p-t_k^{p,j}}_2^2\Bigg)\nonumber\\
&+\lambda_b\:(\sum_{k=1}^K\sum_{p:\{p\in \mathcal{P}_k,b \in \mathcal{U}_k^p\}}t_k^p - C_b)-\sum_{k=1}^K\sum_{p:\{p\in \mathcal{P}_k,b \in \mathcal{U}_k^p\}}\beta_k^p t_k^p,\label{eq:lagdec}
\end{align}
where $\mathbf{t}_b=\{t_k^p\}_{p\in\mathcal{P}_k,b\in \mathcal{U}_k^p,k=1:K}$ and  $\boldsymbol{\beta}_b=\{\beta_k^p\}_{p\in\mathcal{P}_k,b\in \mathcal{U}_k^p,k=1:K} \geq \mathbf{0}$. 
To develop an algorithm to solve each subproblem with respect to $\mathbf{t}_b$, we use KKT conditions. We write the first-order optimality conditions with respect to $\mathbf{t}$ as follows:
\begin{subequations}
	\begin{align}
	&\frac{\partial L_{t,b}(\mathbf{t}_b,\lambda_b,\boldsymbol{\beta}_b)}{\partial t_k^p}=\theta_k\int_{0}^{\underline{r}_k^{p,i+1}} \frac{\partial z_k^p(v_k^p,t_{k}^{p})}{\partial t_k^p}\:(\underline{r}_k^{p,i+1}-v_k^p)dv_k^p\nonumber\\
	&+\lambda_b+\theta_k\zeta_k^{p,j}\:(t_k^p-t_k^{p,j})-\beta_k^p=0,\label{eq:findt1}\\
	&\sum_{k=1}^K\sum_{p:\{p\in \mathcal{P}_k,b \in \mathcal{U}_k^p\}}t_k^p \leq C_b,\label{eq:findt4}\\
	&\lambda_b(\sum_{k=1}^K\sum_{p:\{p\in \mathcal{P}_k,b \in \mathcal{U}_k^p\}}t_k^p -C_b)=0,\:\:\lambda_b\geq 0,\label{eq:findt2}\\
	&\beta_k^p t_k^p=0,\:\:t_k^p\geq 0,\:\:\beta_k^p\geq 0\label{eq:findt3}.
	\end{align}
\end{subequations}
From \eqref{eq:findt1}, we observe that a given dual variable $\lambda_b$, which corresponds to AP $b$, identifies the reserved resource $t_k^p$ for all downlinks created by that AP.  The proposed dual algorithm works as follows: implement a bisection search on $\lambda_b$ in the non-negative orthant and find each $t_k^p: t_k^p\geq 0$, which is associated with the AP $b$, from \eqref{eq:findt1} when $\beta_k^p=0$. Continue the bisection search until one $\lambda_b$ is obtained such that for the obtained $\lambda_b$, we have $\sum_{k=1}^K\sum_{p:\{p\in \mathcal{P}_k,b \in \mathcal{U}_k^p\}}t_k^p =C_b$. If there is no such $\lambda_b$, we set $\lambda_b=0$ and solve \eqref{eq:findt1} and \eqref{eq:findt3} without \eqref{eq:findt4}--\eqref{eq:findt2}. Once the optimized variables are obtained, we update $t_k^{p,j}$ and $j=j+1$. We repeat the same process until $\mathbf{t}^j=\{t_{k}^{p,j}\}_{p\in\mathcal{P}_k,k=1:K}$ converges.  As it is explained in Remark \ref{co:2}, after a sequence of upper-bound minimizations and updating the proximal terms in the objective function, a KKT (local stationary) solution to the original problem is obtained. If expected outage terms for downlinks are non-increasing in $\mathbf{t}$, one can show that the successive upper-bound minimization converges to the global minima with respect to $\mathbf{t}$. After $\mathbf{t}^j$ converges, we set $\underline{\mathbf{t}}^{i+1}=\mathbf{t}^j$.

\subsection{The Proposed BCD Algorithm}
To solve the problem in \eqref{eq:sim} to a KKT point, we optimize with respect to two blocks of variables, $\mathbf{r}$ and $\mathbf{t}$, alternatively with the Gauss-Seidel update style. Therefore, if we choose $\mathbf{r}$ to update first, with  $\underline{\mathbf{r}}^{i+1}$, we optimize with respect to $\mathbf{t}$, and then, we update $\underline{\mathbf{t}}^{i+1}$. We keep optimizing with respect to $\mathbf{r}$ and $\mathbf{t}$ alternatively until both blocks converge. The summary of the overall BCD approach is given in Algorithm \ref{al:bcd}.

\begin{algorithm}[t!]
	0.	\textbf{Initialization}  Feasible initializations for $\underline{\mathbf{r}}^0$ and $\underline{\mathbf{t}}^0$, $i=0$\;
	\Repeat{$\norm{\underline{\mathbf{r}}^{i}-\underline{\mathbf{r}}^{i-1}}_2^2+\norm{\underline{\mathbf{t}}^{i}-\underline{\mathbf{t}}^{i-1}}_2^2$ is small enough}{
		1. Apply Algorithm \ref{al:upper} to solve \eqref{eq:sim} and find $\underline{\mathbf{r}}^{i+1}$\;
		2. Solve \eqref{eq:sim} with respect to $\mathbf{t}$ and find $\underline{\mathbf{t}}^{i+1}$\;
		3. $i=i+1$\;
	}     
	\caption{The proposed BCD algorithm to solve \eqref{eq:sim}}
	\label{al:bcd}
\end{algorithm}
\begin{proposition}
	Algorithm \ref{al:bcd} converges to a KKT solution to \eqref{eq:sim}.
\end{proposition}
\begin{proof}
	First, the objective function of \eqref{eq:sim} is continuously differentiable. 
	Second, feasible sets of two blocks of variables are separate in \eqref{eq:sim}. Hence, updating one block of variables does not change the other block. Third, in each iteration of Algorithm \ref{al:bcd}, a KKT solution is obtained. Therefore, according to \cite[Proposition 3.7.1]{bertsekas1999nonlinear}, the proposed (BCD) Algorithm \ref{al:bcd} converges to a KKT solution. 
\end{proof}

\section{Numerical Tests}
\label{sec:sim}
In this section, we demonstrate the performance of our proposed approach against two heuristic algorithms.
\subsection{Simulation Setup}
The considered network for evaluations is shown in Fig. \ref{fig:net}, which includes both the backhaul and radio access parts. A data center is connected to routers of the network through three gateway routers, GW $1$, GW $2$, and GW $3$. The network includes $57$ APs and $11$ network routers. APs are distributed on the X-Y plane and they are connected to each other and routers via wired links. The backhaul network has $162$ links. Wired link capacities are identical in both directions. 
Backhaul link capacities are determined as
\begin{itemize}
	\item Links between the data center and routers: $4$ Gnats/s;
	\item Links between routers: 2 Gnats/s;
	\item Links between routers and  APs: $2$ Gnats/s;
	\item 2-hop to the routers: $400$ Mnats/s;
	\item 3-hop to the routers: $320$ Mnats/s;
	\item 4-hop to the routers: $160$ Mnats/s.
\end{itemize}
The considered paths originate from the data center and are extended toward users. We consider $200$ users are distributed randomly in the same plane of APs; however, they are not shown in Fig. \ref{fig:net}. User AP associations are determined by the highest long-term received power. We consider three wireless connections, which have the highest received power, to serve each user. There are three paths for carrying data from a data center to APs. The distribution of the demand is log-normal:
\begin{align}
d_k \sim \frac{1}{d_k \sigma_k \sqrt{2\pi}}\exp(-\frac{(\ln d_k-\eta_k)^2}{2 \sigma_k^2}).\label{eq:demand}
\end{align} 
In addition, it is assumed that $\eta_k$ is realized randomly from a normal distribution for each user. The power allocations in APs are fixed. The dispensed resource in an AP is bandwidth. The channel between each user and an AP is a Rayleigh fading channel. The CDF of the wireless channel capacity, which is parameterized by the allocated bandwidth $t_k^p$, is given as follows \cite{4411539}:
\begin{align}
Z_k^p(v_k^p,t_k^p)=1-\exp(\frac{1-2^{v_k^p/t_k^p}}{\overline{\text{SNR}_k^p}}),\nonumber
\end{align}
where $\overline{\text{SNR}_k^p}$ is the average SNR. The PDF of the wireless channel capacity is
\begin{align}
z_k^p(v_k^p,t_k^p)=\frac{\ln(2)2^{v_k^p/t_k^p}\exp(\frac{1-2^{v_k^p/t_k^p}}{\overline{\text{SNR}_k^p}})}{\overline{\text{SNR}_k^p} t_k^p}.\label{eq:dist}
\end{align}
Benchmark heuristic algorithms are the single-path and the average-based approaches. In the single-path approach, each user is served through one path from a data center to a user. Moreover, the average-based algorithm only considers the mean of the user demand and the average achievable rate of a downlink. To compare algorithms, with an identical network, we measure the objective function of \eqref{opt:first}, the sum of user expected supportable rates, the aggregate expected outage of downlinks and the amount of traffic that each algorithm  can reserve for users. One datastream is associated with each user. In total, we have $600$ paths in the backhaul. We use \textsc{C} to implement algorithms. 
\begin{figure}
	\centering
	\includegraphics[width=.45\textwidth]{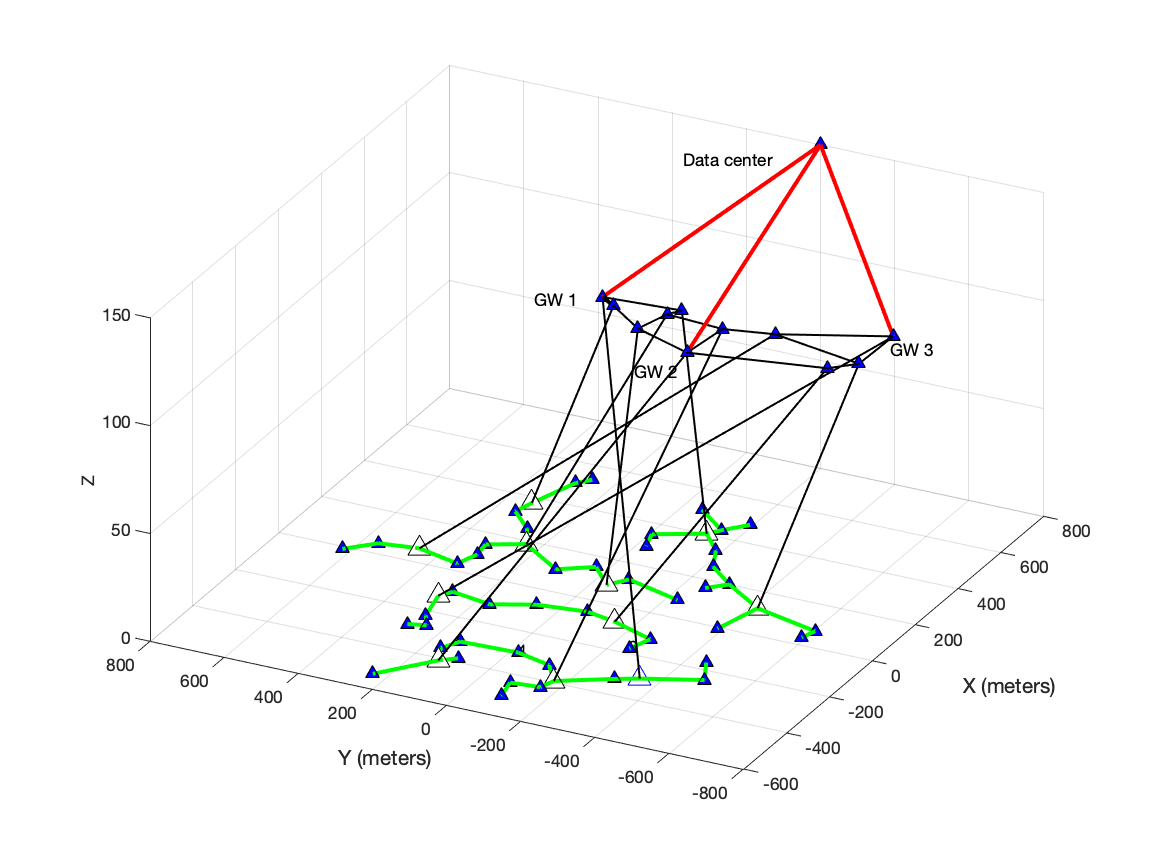}
	\caption{A wireless data network consists of $57$ APs and $11$ routers.}\label{fig:net}
\end{figure} 
\subsection{Learning Probability Density Functions}
	The optimization problem in \eqref{opt:first} takes into account PDFs of user demands and achievable rates of downlinks. When PDFs are not given, one can use a data-driven approach to learn PDFs used in \eqref{opt:first} based on collected observations. Upon the collection of user demands and achievable rates of downlinks, one can estimate the  PDFs using a recursive non-parametric estimator. In order to estimate PDFs in an online streaming fashion, one can use efficient recursive kernel estimators, such as the Wolverton and Wagner estimator \cite{wolverton1969asymptotically}. Suppose that independent random variables $X_1, X_2,\dots,X_n$ are observations that are collected from an identical PDF $\chi$ with respect to Lebesgue's measure. The estimated PDF is 
	\begin{align}
	\hat{\chi}_{n,\mathbf{h}_n}=\frac{1}{n}\sum_{k=1}^{n}\frac{1}{h_k}K(\frac{X_k-x}{h_k}),\nonumber
	\end{align}\normalsize
	where $\mathbf{h}_n=(h_1,h_2,\dots,h_n)$, $h_1>\dots>h_n$ and $K(\cdot)$ is a kernel function. The advantage of the above estimator is that it can be written in a recursive form as follows:
	\begin{align}
	\hat{\chi}_{n+1,\mathbf{h}_{n+1}}=\frac{n}{n+1}\hat{\chi}_{n,\mathbf{h}_n}+\frac{1}{(n+1)h_{n+1}}K(\frac{X_{n+1}-x}{h_{n+1}}),\nonumber
	\end{align}\normalsize
	which makes it suitable for real-time applications. The bandwidth selection in \cite{comte2019bandwidth} can be used for the above estimator. The bandwidth $h_k$ is selected in \cite{comte2019bandwidth} as $h_k=k^{-\gamma}, k\in \{1,\dots,n\}$, where $\gamma=\frac{1}{2\beta+1}$ and $\beta >0$.
	\subsection{Simulation Results}

Before demonstrating the performance of Algorithm \ref{al:bcd}, we depict the convergence of Algorithm \ref{al:parallel} in Fig. \ref{fig:conv}. The convergence of Algorithm \ref{al:parallel} for different means of the user demand is depicted in Fig. \ref{fig:conv}. It is observed that Algorithm \ref{al:parallel} has a fast convergence for the large network of Fig. \ref{fig:net} with $600$ paths. Numerical results show that the number of required iterations for Algorithm \ref{al:upper} to converge for the simulation setting described above is at most $60$. The CPU time for Algorithm \ref{al:upper} is measured and is given in Table \ref{eq:cpu}. 

\begin{figure}
	\centering
	\includegraphics[width=.45\textwidth]{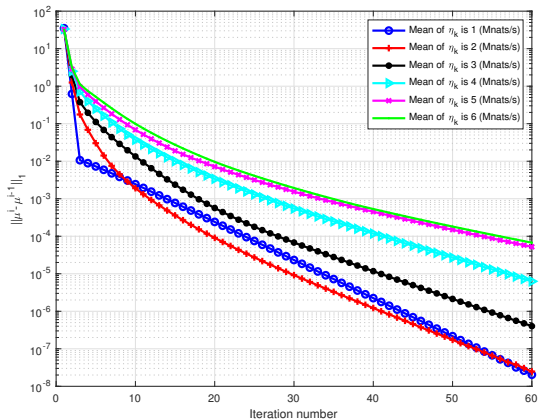}
	\caption{The convergence of Algorithm \ref{al:parallel}.}\label{fig:conv}
\end{figure}
\begin{table}
	\centering
	\caption{\textsc{CPU time for $60$ iterations of Algorithm 3.}}
	\begin{tabular}{lllllll}
		\hline \hline 
		Mean of $\eta_k$ & $1$ Mnats/s & $2$ Mnats/s & $3$ Mnats/s & $4$ Mnats/s  &      \\
		CPU time & $0.051$ s & $0.082$ s & $0.104$ s & $0.121$ s \\
		Mean of $\eta_k$ & & $5$ Mnats/s & $6$ Mnats/s &      \\
		CPU time  & & $0.144$ s & $0.158$ s   \\
		\hline
	\end{tabular}\label{eq:cpu}
	\vspace{1ex}
\end{table}

\begin{figure}
	\centering
	\begin{subfigure}[b]{0.48\textwidth}
		\includegraphics[width=\textwidth]{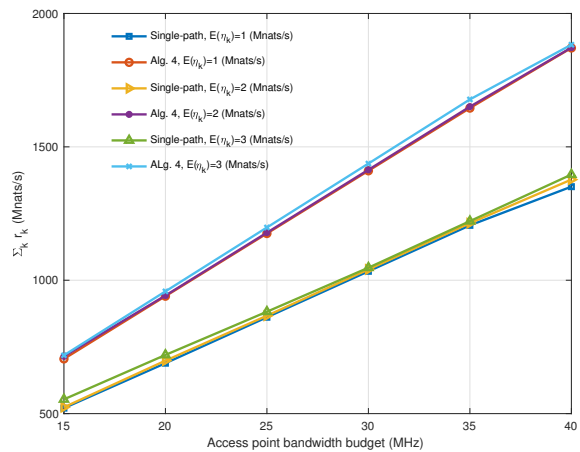}
		\caption{}
		\label{fig:se_opt}
	\end{subfigure}
	\begin{subfigure}[b]{0.48\textwidth}
		\includegraphics[width=\textwidth]{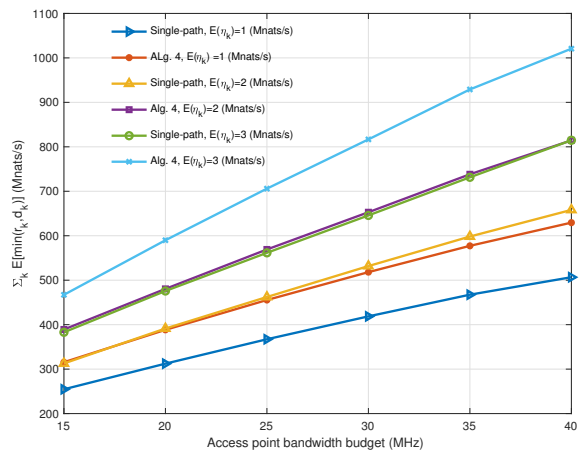}
		\caption{}
		\label{fig:se_obj}
	\end{subfigure}
	\begin{subfigure}[b]{0.48\textwidth}
		\includegraphics[width=\textwidth]{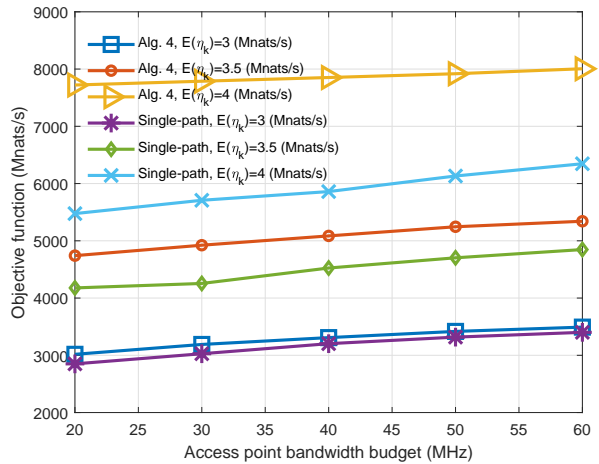}
		\caption{}
		\label{fig:obj}
	\end{subfigure}
	\caption{(a) Reserved rates by Algorithm \ref{al:bcd} (multi-path) and the single-path approach when wireless channels are deterministic.  (b) The expected supportable rates for users by Algorithm \ref{al:bcd} and the single-path approach when wireless channels are deterministic. (c) The objective function of problem \eqref{opt:first} with the single-path approach and Algorithm \ref{al:bcd} when wireless channels are stochastic.}
\end{figure}

\begin{figure*}
	\centering
	\begin{subfigure}{0.48\textwidth}
		\includegraphics[width=\textwidth]{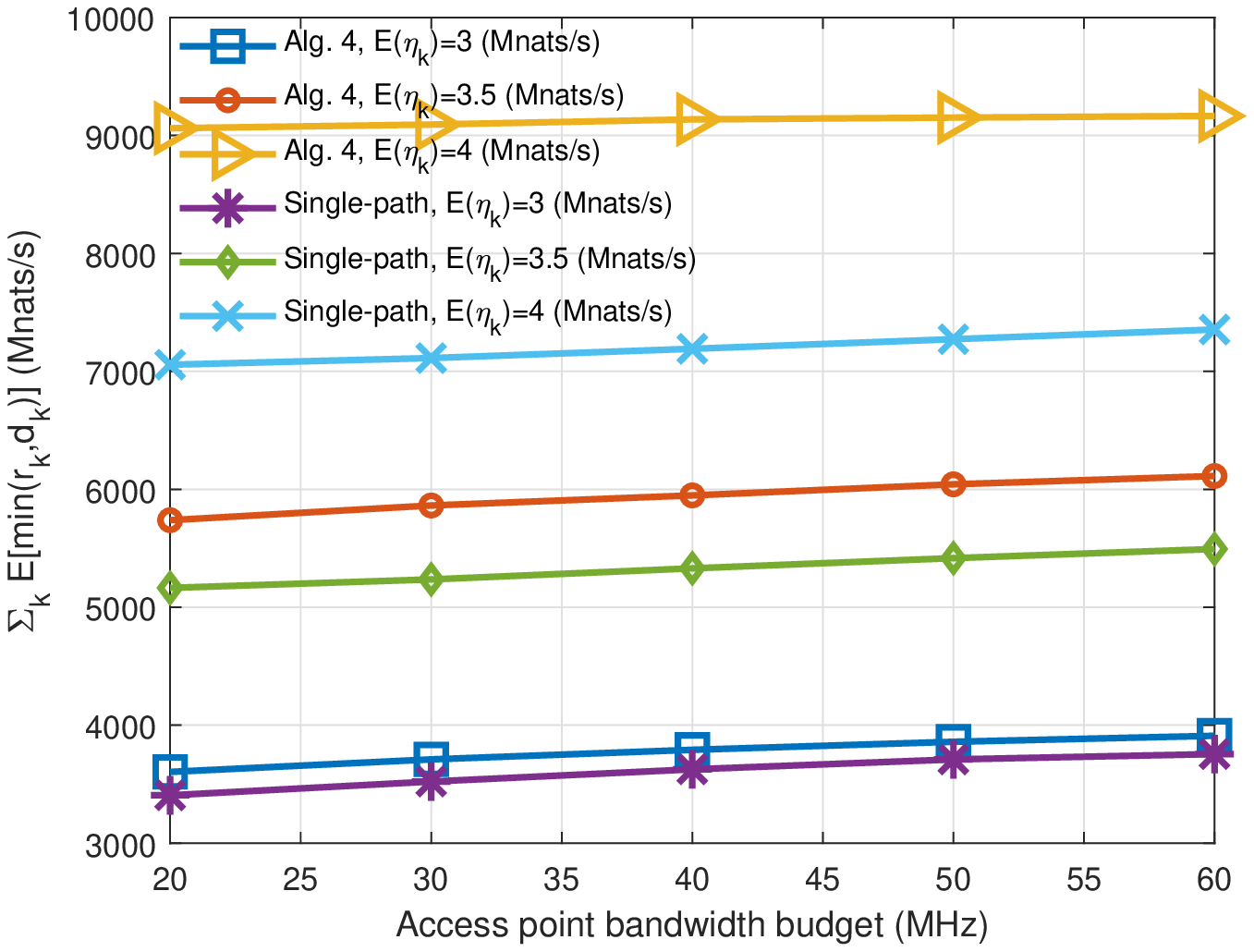}
		\caption{}
		\label{fig:demand}
	\end{subfigure}
	\begin{subfigure}{0.48\textwidth}
		\includegraphics[width=\textwidth]{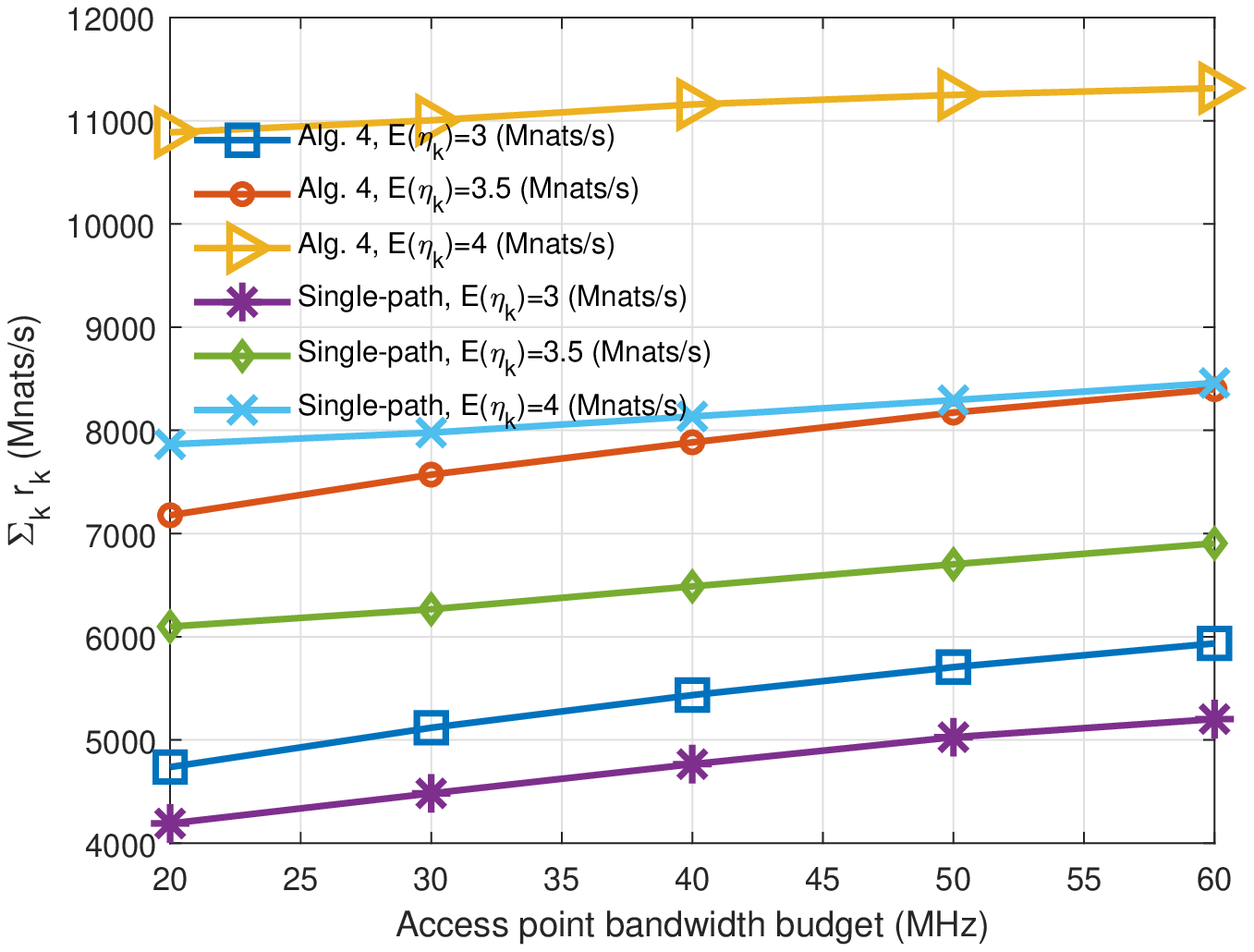}
		\caption{}
		\label{fig:r}
	\end{subfigure}
	\caption{Stochastic wireless channels: performance of the single-path approach and Algorithm \ref{al:bcd} in terms of (a) the aggregate expected supportable traffic; and (b) aggregate reserved rates.}
\end{figure*}
\begin{figure*}
	\centering
	\begin{subfigure}[b]{0.48\textwidth}
		\includegraphics[width=\textwidth]{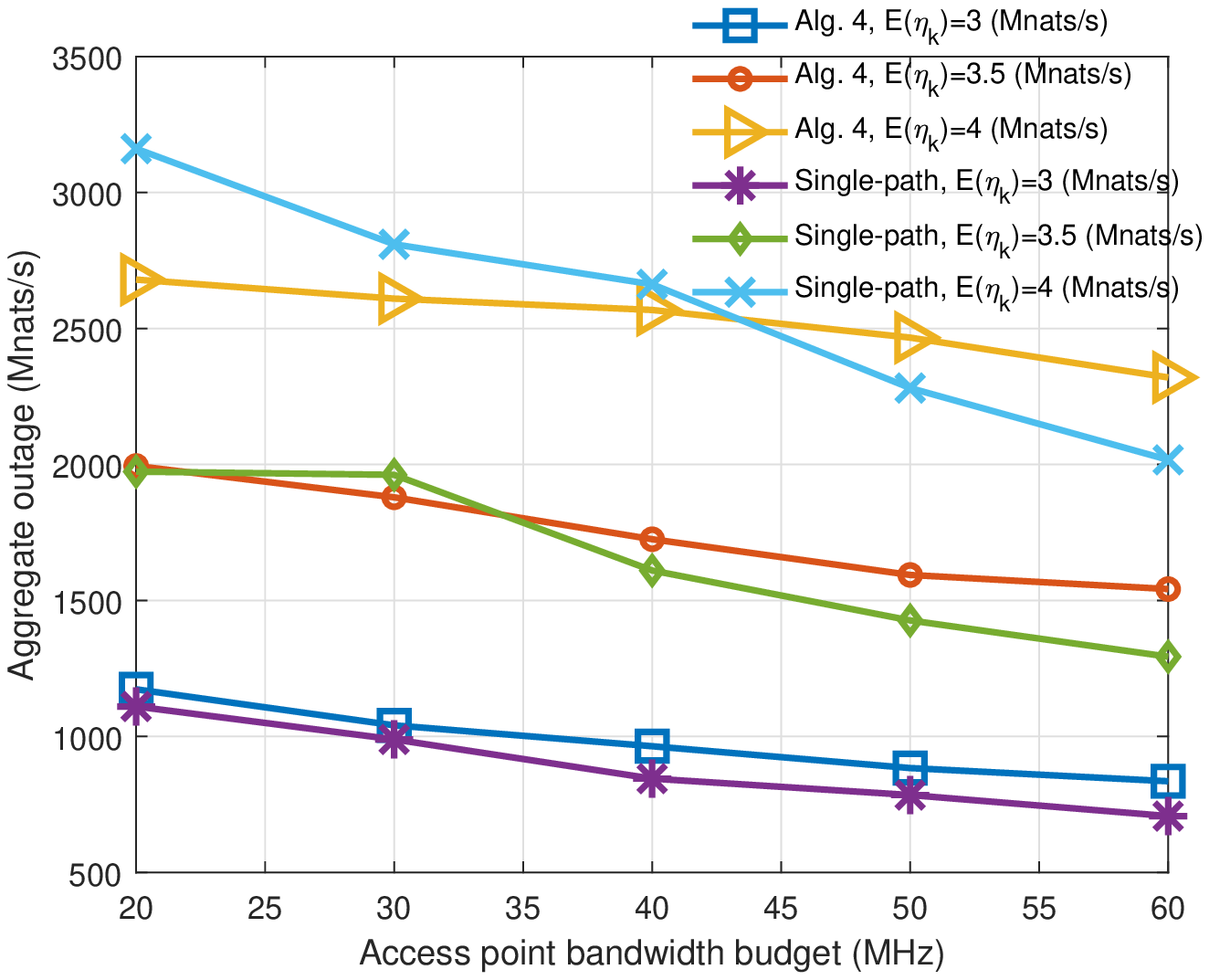}
		\caption{}
		\label{fig:outage}
	\end{subfigure}
	\begin{subfigure}[b]{0.48\textwidth}
		\includegraphics[width=\textwidth]{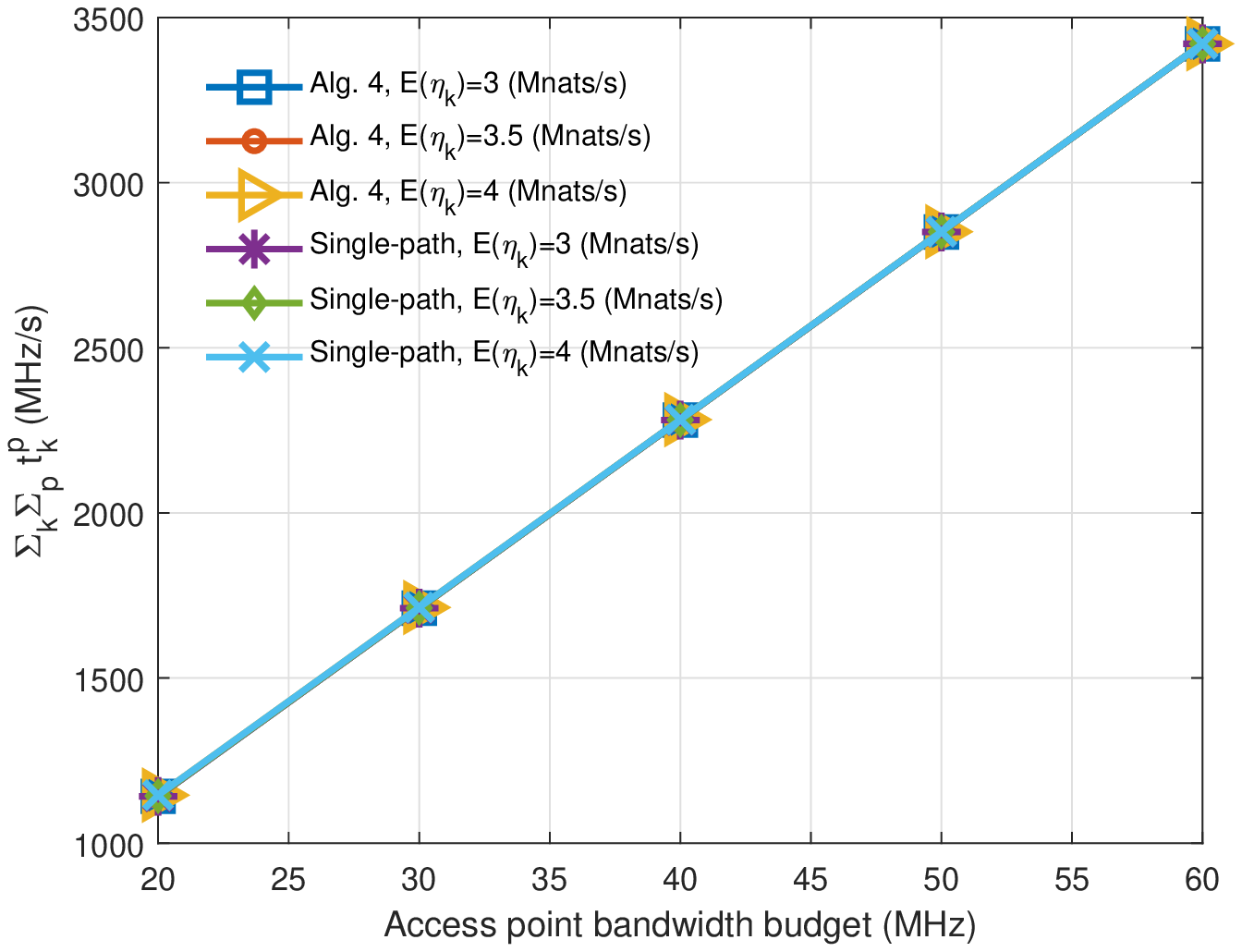}
		\caption{}
		\label{fig:t}
	\end{subfigure}
	\caption{ Stochastic wireless channels: performance of the single-path approach and Algorithm \ref{al:bcd} in terms of (a) expected outage of downlinks; and (b) reserved bandwidth.}
\end{figure*}
First, let us assume that transmission rates on downlinks are deterministic functions of bandwidth in APs. Therefore, no outage (rate loss) is considered. For each downlink, the transmission rate and the allocated bandwidth are connected to each other as $r_k^p=\delta_k^p t_k^p$, where $\delta_k^p$ is the spectral efficiency of the downlink of path $p$ to serve user $k$. Furthermore, suppose that  $\sigma_k=3.6$ and the capacity of each backhaul link listed previously is divided by $4$.  When the bandwidth budget of each AP increases from $15$ MHz to $40$ MHz, the aggregate reserved rates for users by Algorithm \ref{al:bcd} (multi-path) and the single-path approach are shown in Fig. \ref{fig:se_opt}. The aggregate expected supportable rates of users with both approaches are depicted in Fig. \ref{fig:se_obj}. It is observed that Algorithm \ref{al:bcd} outperforms the single-path approach. Both approaches utilize all available bandwidth in APs.

Consider the distribution of each wireless channel (downlink) achievable rate follows \eqref{eq:dist} and backhaul link capacities are as listed previously. Suppose that the available bandwidth in each AP increases by a step size of $10$ MHz, where $\theta_k=1/2$ and $\sigma_k=0.6$.  The objective function of the problem in \eqref{opt:first} by Algorithm \ref{al:bcd} and the single-path approach are compared in Fig. \ref{fig:obj}. Our proposed Algorithm \ref{al:bcd} outperforms the single-path approach. It is observed that with the increase of mean for $\eta_k$ and the AP bandwidth budget, the objective function increases. 

The expected supportable demands of users, depicted in Fig. \ref{fig:demand}, increases when the mean of $\eta_k$ and the AP bandwidth budget increase. 
It is observed from Fig. \ref{fig:demand} that the aggregate expected supportable traffic for users obtained by Algorithm \ref{al:bcd} is greater than that by the single-path approach. In Fig. \ref{fig:r}, we observe that the aggregate reserved rates for users increases with the increase of mean for $\eta_k$. Furthermore, it increases when the bandwidth budgets of APs increase. From Fig. \ref{fig:outage}, we observe that the aggregate expected outage increases as the mean of $\eta_k$ increases and decreases when the AP bandwidth budget increases. We observe from Fig. \ref{fig:t} that the bandwidth reservation by Algorithm \ref{al:bcd} is almost equal to that by the single-path approach. Numerical results show that $5$ iterations are sufficient for the convergence of Algorithm \ref{al:bcd}.

Next, we evaluate the performance of Algorithm \ref{al:bcd} against the average-based approach when both the demand and downlink achievable rates are stochastic. The average-based algorithm is oblivious to the user demand and the downlink achievable rate distributions. It only considers the average of each user demand and the average achievable rate of a downlink. The average-based approach uses the same set of paths used by Algorithm \ref{al:bcd}. The bandwidth budget in each AP is 40 MHz. Furthermore, $\sigma_k=0.6$ and $\theta_k=1/3$. The demand and downlink achievable rate distributions are as given in \eqref{eq:demand} and \eqref{eq:dist}, respectively. Both approaches are set to make reservations for users assuming the mean of $\eta_k$ is $2$ Mnast/s. We generate $100$ scenarios in which user demands and downlink capacities are random. For each scenario, we measure how much the user demands are satisfied using the reserved resources in the network by both approaches. After collecting results for $100$ scenarios, we plot the empirical CDF for the supply demand ratio in Fig. \ref{fig:fading2}. It is observed that when the mean of demand exceeds what it was supposed to be, the resource reservation made by Algorithm \ref{al:bcd} is more robust and supports random demands better. The total reserved link capacities in the backhaul by Algorithm \ref{al:bcd} is $1.0979\times 10^4$ Mnats/s and is $7.74\times 10^3$ Mnats/s by the average-based approach. Furthermore, the total reserved bandwidth in RAN by Algorithm \ref{al:bcd} is $2.043\times 10^3$ MHz and is $1.968\times 10^3$ MHz by the average-based approach.
\begin{figure}
	\centering
		\includegraphics[width=0.5\textwidth]{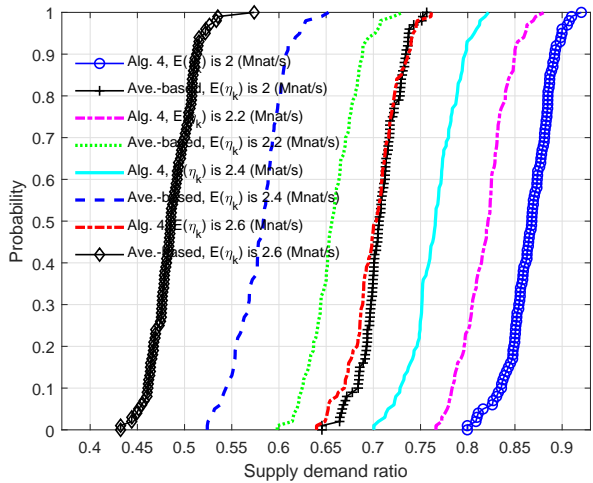}
	\caption{ The probability of being able to support the user demands up to a certain percentage.  }\label{fig:fading2}
\end{figure}
\\
\section{Concluding Remarks and Future Directions}
In this paper, we studied link capacity and transmission resource reservation in wireless data networks prior to the observation of user demands. Using the statistics of user demands and achievable rates of downlinks, we formulated an optimization problem to maximize the sum of user expected supportable traffic while minimizing the expected outage of downlinks. We demonstrated that this problem is non-convex in general. To solve the problem approximately, an efficient BCD approach is proposed which benefits from distributed and parallel computation when each block of variables is chosen to be updated. We demonstrated that despite the non-convexity of the problem, our proposed approach converges to a KKT solution to the problem. We verified the efficiency and the efficacy of our proposed approach against two heuristic algorithms developed for joint resource reservation in the backhaul and RAN.

In future work, we consider multi-tenant networks and reservation-based network slicing. In addition to users, tenants have different requirements \cite{Reyhanian4}, and maximum isolation between sliced resources should be enforced \cite{Reyhanian2}. The demand distribution of users may change over time and the network resources should be sliced for tenants accordingly. However, the slice reconfiguration for each tenant involves cost and overhead. Based on the cost of reconfiguration and newly arrived statistics, we formulate the problem from a sparse optimization perspective and propose an efficient approach based	on iteratively solving a sequence of group Least Absolute Shrinkage and Selection Operator (LASSO) problems \cite{Reyhanian4}.

\ifCLASSOPTIONcaptionsoff
\newpage
\fi

\bibliographystyle{IEEEbib}
\bibliography{ref_SDRA}

\appendix[Block Successive Upper-Bound Minimization]
Notations in this Appendix are identical to \cite{razaviyayn2013unified} and are not related to those defined in the paper. According to the BSUM algorithm \cite[Theorem 2]{razaviyayn2013unified}, when an upper-bound satisfies four conditions, the solution acquired by the BSUM converges to a local minima to the problem. Here, we give a brief description of the BSUM approach. 
Suppose that $u_i(\mathbf{x},\mathbf{x}^{t-1})$ is an upper-bound for an arbitrary objective function $f(\mathbf{x})$ at the point $\mathbf{x}^{t-1}$. In iteration $t$, one selected block (say, block i) is optimized by solving the following subproblem:
\allowdisplaybreaks
\begin{equation}
\begin{aligned}
& \underset{\mathbf{x}_i}{\min}
& & u_i(\mathbf{x}_i,\mathbf{x}^{t-1})\label{opt:bsum}\\
& \text{s.t.}
& & \mathbf{x}_i \in \mathcal{X}_i,
\end{aligned}
\end{equation} 
where $\mathcal{X}_i$ is the feasible set of block $\mathbf{x}_i$.
Conditions on the upper-bound are listed in \cite[Assumption 2]{razaviyayn2013unified} as follows:
\begin{enumerate}
	\item $u_i(\mathbf{y}_i,\mathbf{y})=f(\mathbf{y}),~ \forall \mathbf{y}\in \mathcal{X}, \forall i,$
	\item $u_i(\mathbf{x}_i,\mathbf{y})\geq f(\mathbf{y}_1,\dots,\mathbf{y}_{i-1},\mathbf{x}_i,\mathbf{y}_{i+1},\dots,\mathbf{y}_n),~ \forall \mathbf{x}_i\in \mathcal{X}_i, \forall \mathbf{y}\in \mathcal{X}, \forall i,$
	\item $u_i^{'}(\mathbf{x}_i,\mathbf{y};\mathbf{d}_i)|_{\mathbf{x}_i=\mathbf{y}_i}= f'(\mathbf{y};\mathbf{d}),\: \forall \mathbf{d}=(0,\dots,\mathbf{d}_i,\dots,0)\:\: \text{s.t.}\:\: \mathbf{y}_i+\mathbf{d}_i\in \mathcal{X}_i, \forall i,$
	\item $u_i(\mathbf{x}_i,\mathbf{y})$ is continuous in $(\mathbf{x}_i,\mathbf{y}), \forall i$.
\end{enumerate}
When problem \eqref{opt:bsum} is solved sequentially for different $i$ and there exists a unique solution for each subproblem, $\mathbf{x}^t$ converges to a KKT point of $f(\mathbf{x})$.
\end{document}